%% file: main.tex
\documentclass[a4paper,10pt,onecolumn, unpublished] {quantumarticle}
\pdfoutput=1
\usepackage[utf8]{inputenc}
\usepackage[english]{babel}
\usepackage[T1]{fontenc}
\usepackage{amsmath}
\usepackage{hyperref}
\usepackage[numbers]{natbib}
\usepackage{tikz}
\usepackage{lipsum}
\usepackage{physics}
\usepackage{caption}
\usepackage{subcaption}
\usepackage{amsthm}
\newtheorem{definition}{Definition}
\newtheorem{theorem}{Theorem}

\newtheorem{remark}{Remark}

\begin{document}

\title{Quantum ramp secret sharing from Haar scrambling}

\author{Kiran Adhikari}
\email{kiran.adhikari@tum.de}
\affiliation{Emmy Noether Group for Theoretical Quantum Systems Design, Technical University of Munich, Germany}

\maketitle

\begin{abstract}
Quantum information scrambling has emerged as a powerful tool for studying the dynamics of chaotic quantum many-body systems, assessing benchmarking protocols, and even investigating exotic black hole models \cite{r10,r11}. During quantum information scrambling, localized quantum information disperses across the entire system, hiding from the observers who can only access part of it. On the other side, we have a  fundamental cryptographic primitive called secret-sharing schemes, where a dealer shares a quantum secret with a group of parties, such that any subset of parties above a specific threshold size can reconstruct it. In this paper, we demonstrate that two protocols, Haar scrambling, which serves as a baseline for other scrambling techniques, and quantum secret sharing, are indeed equivalent.  The scheme one gets out of this is of a ramp secret-sharing nature rather than a threshold scheme. We expect this because of the inherent randomness present in Haar unitaries. Moreover, by varying the purity of the initial states, we demonstrate that it is possible to obtain all possible types of ramp secret-sharing schemes. 
Furthermore, utilizing complexity theoretic arguments, we argue that the protocol can be implemented efficiently, i.e., in polynomial time. Finally, we explore potential applications, ranging from security implications for distributed quantum networks to cryptography in the Noisy Intermediate-Scale Quantum (NISQ) era, and draw insights for fundamental physics, such as many-body physics and quantum black holes. 
\end{abstract}

\input{content}

\section{Acknowledgements}
The research is part of the Munich Quantum Valley, which is supported by the Bavarian state government with funds from the Hightech Agenda Bayern Plus.

\onecolumn
\appendix

\input{appendix}

\bibliographystyle{quantum}
\bibliography{references}
\end{document}

%% file: content.tex
\section{Introduction}

Quantum information scrambling has emerged as a crucial tool for studying complicated quantum systems \cite{Landsman:2018jpm, Schuster:2021uvg, r11} ranging from chaotic many-body quantum systems to exotic quantum black hole physics \cite{Sekino_2008, Hayden_2007} and facilitating various randomization and benchmarking protocols \cite{ doi:10.1126/science.abg5029, Blok_2021}. The process of quantum scrambling leads to the generation of global entanglement, causing initially localized quantum information to disperse throughout the system, thereby concealing it from observers who can only access part of it. Therefore, this phenomenon can also be seen as a form of "thermalization of the quantum information" and is crucial in understanding information spreading in quantum systems.

The literature on quantum information scrambling encompasses diverse forms and definitions, which can broadly be classified into three types. The first family emphasizes Hamiltonians exhibiting quantum chaotic features with out-of-time ordered correlators (OTOCs) being a popular diagnostic tool \cite{larkin, Maldacena_2016,Hashimoto_2017}. In this context, the decay of OTOCs serves as a signature of quantum information scrambling. The second family focuses on information-theoretic tools, such as mutual information and tri-partite information, to track the spread of information and, therefore, to define quantum information scrambling \cite{r11}. However, obtaining exact analytical results for these information theoretic quantities is challenging. The third family of models, where exact analytical results can be obtained, involves random unitaries drawn from the Haar measure, which is called Haar scrambling \cite{Sekino_2008}. Haar scrambling serves as a baseline for scrambling, as we expect the late-time values of entropy and mutual information for any scrambling unitary to match those of Haar scrambling. Furthermore, Haar scrambling falls under the category of chaotic scrambling, as the out-of-time ordered correlators (OTOCs) for Haar random unitaries exhibit characteristic decay behavior.

From the realm of many-body quantum systems, let us now traverse into the fascinating land of cryptography. Secret-sharing schemes serve as a fundamental primitive for various cryptographic protocols. In a threshold quantum secret sharing protocol \cite{r9}, a dealer shares a quantum secret with a group of parties such that any subset of parties above a specific threshold size can reconstruct it. In contrast, parties below the threshold size are not allowed to. Shamir \cite{shamir} and Blakley \cite{blakley} proposed the first classical secret sharing schemes, while various quantum secret sharing schemes were proposed in \cite{r7, Hillery_1999, r8, r9}. In the literature, various types of secret sharing schemes can be found. In a quantum threshold scheme $((k,n))$, an arbitrary quantum secret is divided into $n$ parties such that any set of parties larger than $k$ can perfectly reconstruct it while parties of size less than $k$ have no information about it \cite{r7}. Unlike in threshold secret sharing schemes, there also exists a ramp secret sharing protocol, where partial information about the secret is allowed to leak to a set of participants that cannot fully reconstruct the secret. In particular, there are parameters $1 \leq b < g \leq n$ such that the set of parties of size at least $g$ should be able to reconstruct the secret,  while sets of parties of size at most $b$ cannot learn anything about the secret. Furthermore, we make no requirement on the parties of size $g-b$. Blakley \cite{blakley} and Meadows \cite{ramp_scheme}  first introduced the classical ramp secret sharing scheme, which they later used to construct efficient and secure multiparty computation protocols \cite{ramp_scheme_2}. Compared to quantum threshold secret sharing, quantum ramp secret sharing has not been extensively explored \cite{Zhang_2014, MATSUMOTO_2018}. 

Intuitively, classical cryptography can be viewed as the scrambling and descrambling (encryption and decryption) of messages between two parties. Our goal in this paper is to have a quantum version of it. We primarily focus on Haar scrambling for two reasons. Firstly, it allows us to derive analytical results. Secondly, Haar scrambling serves as a baseline for other scrambling unitaries, enabling the extension of results obtained for Haar scrambling to other scenarios as well. A Haar scrambling unitary has the capacity to generate global entanglement, which implies that they can be a good encoder for quantum error correction codes, as implied by the decoupling theorem \cite{decoupling}.  Given that all quantum secret sharing schemes are variants of quantum error correcting codes \cite{r7, r8, Marin_2013}, it is natural to anticipate a connection between scrambling and secret sharing. In particular, bounds from standard as well as entanglement-assisted quantum error-correcting codes \cite{Grassl_2022, PhysRevA.76.062313} can be used to create various kinds of secret-sharing schemes. 

In this paper, we show that two protocols, Haar scrambling, and quantum secret sharing, are indeed equivalent. The scheme one gets out of this is of a ramp secret-sharing nature rather than a threshold scheme. The ramp nature is expected because of the inherent randomness in Haar scrambling. By changing the purity of the initial states, we show that it is possible to obtain all possible ramp secret-sharing schemes. Furthermore, the protocol can be efficiently implemented in polynomial time by approximating the scrambling unitary using $t$-design circuits. We also discuss some potential applications of our work, such as security applications in a distributed quantum network, cryptography in the Noisy Intermediate-Scale Quantum (NISQ) era \cite{Preskill_2018}, and drawing insights for fundamental physics domains, including many-body physics and quantum black holes. 
The structure of the paper is as follows. In Section \ref{sec:preliminaries}, we review concepts from quantum information theory relevant to the paper. In section \ref{sec:quantumInformationScrambling}, we review various notions of quantum information scrambling introduced in the literature. We also propose a novel definition of information scrambling based on quantum information theoretic quantities. In section \ref{sec:RampSecret}, we formalize the notion of ramp secret sharing and show that it is equivalent to Haar scrambling.
\section{Preliminaries}
\label{sec:preliminaries}
 In this section, we provide a brief review of the concepts from quantum information theory relevant to this paper. For a more comprehensive introduction to these concepts, interested readers can refer to the literature \cite{qcn,r6}. Let us consider a quantum system $X$ consisting of $m$ degrees of freedom modeled by the algebra of $m \times m$ matrices over the complex numbers, denoted by $\mathcal{M}_m$. The state of $X$ is described by its density matrix $\rho_X \in \mathcal{M}_m$. In classical information theory, Shannon entropy provides a measure of uncertainty of a random variable, whereas von Neumann entropy generalizes it to the quantum case. The von Neumann entropy of $X$ with density matrix $\rho_X$ is given by:
\begin{equation}
S(X) = -\operatorname{Tr} (\rho_X \log \rho_X) = -\sum_{j=1}^m \lambda_j \log \lambda_j,
\end{equation}
where $\lambda_1, \dots, \lambda_m$ are the eigenvalues of $\rho_X$. Using quantum entropy, one can introduce a quantum generalization of classical mutual information, called quantum mutual information, which gives a measure of the correlation between different quantum systems.  Let us consider two quantum systems, $X$ and $Y$, with density matrices $\rho_X$ and $\rho_Y$, respectively. The quantum mutual information between $X$ and $Y$ is defined as:
\begin{equation}
I(X:Y) = S(X) + S(Y) - S(XY),
\end{equation}
where $S(X)$, $S(Y)$, and $S(XY)$ are the von Neumann entropies of $X$, $Y$, and the joint system $XY$, respectively. In addition to the von Neumann entropy, there exists a quantum version of the classical R\'enyi entropy, called the quantum R\'enyi entropy. For a given quantum system $X$ with density matrix $\rho_X$, the R\'enyi entropy of order $\alpha$ is defined as:
\begin{equation}
S^\alpha(X) = \frac{1}{1-\alpha} \log \text{Tr}[\rho_X^\alpha],
\end{equation}
where $\alpha \geq 0$ and $\alpha \neq 1$. For integer values of $\alpha$, the R\'enyi entropy only involves integer powers and traces of the density matrices, making it easier to compute for physically relevant situations. Moreover, in the limit $\alpha \rightarrow 1$, the R\'enyi entropy reduces to the von Neumann entropy:
\begin{equation}
\lim_{\alpha \rightarrow 1} S^\alpha(X) = S(X).
\end{equation}
The R\'enyi entropy is upper bounded by $\log d$, where $d$ is the dimension of the Hilbert space of the system, and this upper bound is achieved if and only if $\rho_X$ is a maximally mixed state, i.e., $\rho_X = \frac{\mathbf{1}}{d}$. We can also define the R\'enyi version of mutual information between two quantum systems $X$ and $Y$ as:
\begin{equation}
I^\alpha(X:Y) = S^\alpha(X) + S^\alpha(Y) - S^\alpha(XY).
\end{equation}

Frequently, there is a need to compare two different quantum states. Such a comparison can aid in, for instance, evaluating the output of an ideal quantum protocol against a noisy one. One possible method is to employ the operator trace norm, also referred to as the $L_1$ norm. The $L_1$ norm for any operator $M$ is defined as follows:

\begin{equation}
||M||_1 = \text{Tr} \sqrt{M^\dagger M}.
\end{equation}
This norm can be used to calculate the distance between two density matrices $\rho$ and $\sigma$, which is represented by $||\rho - \sigma||_1$. The trace norm distance is motivated physically because for $||\rho - \sigma ||_1 < \epsilon$, $\text{Tr}(\Pi (\rho - \sigma)) < \epsilon$ for any projection operator $\Pi$. This implies that the probability outcome of any experiment between two density matrices $\rho$ and $\sigma$ differs by at most $\epsilon$. Once the distance between two density matrices is established, the quantum mutual information can be employed to place an upper bound on the correlation between two distinct quantum systems. This bound is also referred to as quantum Pinsker's inequality \cite{r6}. The normalized trace distance between $\rho(XY)$ and $\rho(X) \otimes \rho(Y)$ can be defined as $\Delta = \frac{1}{2} ||\rho(XY) - \rho(X) \otimes \rho(Y)||_1$. Then quantum Pinsker's inequality implies \cite{r6}:
\begin{equation}
    \frac{2}{\ln 2} \Delta^2 \leq I(X:Y)
\end{equation}

When the mutual information $I(X:Y)$ is small, it indicates little correlation between the two quantum systems $X$ and $Y$. Thus the joint state $\rho(XY)$ can be approximated by the tensor product of the marginal states $\rho(X)$ and $\rho(Y)$. Another measure of distance between two quantum states is the fidelity, which is defined as $F(\rho(X), \rho(Y)) = (\text{Tr} \sqrt{\rho(X)\rho(Y)})^2$, where $\sqrt{M}$ denotes the unique positive square root of a positive semidefinite matrix $M$. The trace distance can be used to provide upper and lower bounds on the fidelity through the Fuchs-van de Graaf inequalities \cite{r6}, which are given by 
\begin{equation}
    1 - \sqrt{  F(\rho (X),\rho (Y) )} \leq \Delta  \leq \sqrt{1 - F(\rho (X),\rho (Y) ) }
\end{equation} where $\Delta$ is the normalized trace distance. The fidelity is bounded below by $0$ and above by $1$. Two density matrices $\rho(X)$ and $\rho(Y)$ are said to be $\mu$-distinguishable if $F(\rho(X), \rho(Y)) \leq 1-\mu$, and $\nu$-indistinguishable if $F(\rho(X), \rho(Y)) \geq 1-\nu$.


Another extremely important concept in quantum information theory is the decoupling inequality, which has wide-ranging applications such as quantum error correction, transmission of quantum information, protocols such as state merging, coherent state merging, thermodynamics, many-body physics, and even black hole information theory \cite{Hayden_2007, decoupling, Grassl_2022, Sekino_2008}. Suppose a composite system $XY$ is in a joint state $\rho(XY)$. The subsystem $X$ is decoupled from $Y$ if: $\rho(XY) = \rho(X) \otimes \rho(Y)$. This means that the system $X$ is uncorrelated with $Y$ and is therefore statistically independent of any measurement on $Y$. One common way to analyze this setup is to start with the pure state system $RA$, Alice's message $A$ purified with reference $R$. After this, we encode Alice's message $A$ using a unitary encoding and allow it to interact with the environment, which yields the final output $BE$. $B$ here can represent Bob's received message. If the encoding is done such that $R$ and $E$ decouple, then this implies that there exists a decoding protocol, allowing Bob to recover Alice's message $A$ just by having system $ B$. This works even for the case when $R$ and $E$ are approximately decoupled, $\rho(RE) \stackrel{\epsilon}{\approx} \rho(R) \otimes \rho(E)$, that is,
\begin{equation}
F(\rho(RE),\rho (R) \otimes \rho (E)) \geq 1 - \epsilon \text{ for some } \epsilon \geq 0
\end{equation}
Uhlmann's theorem \cite{r6} guarantees that there exists a decoder such that having system $B$, one can approximately reconstruct Alice's message $\rho (A') \stackrel{\epsilon}{\approx} \rho (A)$, that is:
\begin{equation}
F(\rho (A'), \rho (A)) \geq 1 - \epsilon \text{ for some } \epsilon \geq 0
\end{equation}
We will later utilize the decoupling theorem to define the concept of quantum information scrambling.

\subsection{Secret sharing schemes}
\label{sec:secretSharing}
 Secret-sharing schemes serve as cryptographic primitives and are commonly employed in the construction of various cryptographic protocols. Classical secret sharing schemes were proposed by Shamir \cite{shamir} and Blakley \cite{blakley}, while various quantum secret sharing schemes were proposed in \cite{r7, Hillery_1999, r8, r9}. In a secret sharing scheme protocol, a dealer $\mathcal{D}$ distributes a secret $S$ among a set of players $\mathcal{P}$ such that only a certain sufficiently large group can reconstruct it \cite{r9}. Let $\mathcal{P}= \{P_1....P_m\}$ be a set of players. A collection $\Gamma  \subseteq 2^{\{P_1....P_m\}}$ is monotone if $B \in \Gamma$ and $B \subseteq C$ imply that $C \in \Gamma$. An access structure $\Gamma= (\Gamma_{\text{YES}},\Gamma_{\text{NO}})$ is a pair of collections of sets such that $\Gamma_{\text{YES}}, \Gamma_{\text{NO}} \subseteq 2^{\{P_1....P_m\}} $, the collections $\Gamma_{\text{YES}}$ and $ 2^{\{P_1....P_m\}} \setminus \Gamma_{\text{NO}}$ are monotone and $\Gamma_{\text{YES}} \cap \Gamma_{\text{NO}} = \emptyset$. Sets in $\Gamma_{\text{YES}} $ are called authorized sets, and the sets in $\Gamma_{\text{NO}}$ are called unauthorized sets. The access structure is called an incomplete access structure if there is at least one subset of parties $A \subseteq P$ such that $A \not \in \Gamma_{\text{YES}} \cup \Gamma_{\text{NO}}$. Otherwise, it is referred to as a complete access structure.

In a quantum secret sharing scheme, a dealer $D$ wants to share a quantum state $\ket{X}$ with a set of players $\mathcal{P}$ with access structure $\Gamma \in 2^\mathcal{P}$. The quantum secret $\ket{X}$ is choosen from set of possible quantum set $\mathcal{X} = \{ \ket{X_1},\ket{X_2},....,\ket{X_n} \} $, where $\ket{X_i}$ is chosen with probability $p_i$.
Then, the state of the quantum secret can be represented by the density matrix $\rho_S$ as:
\begin{equation}
    \rho_S = p_1 \ket{X_1} \bra{X_1} + p_2 \ket{X_2}\bra{X_2} + ........ + p_n \ket{X_n}\bra{X_n} 
\end{equation}
 Let us introduce a reference system $R$ with Hilbert space $\mathcal{H}_R$, and a purification denoted by $\ket{SR} \in \mathcal{H}_S \otimes \mathcal{H}_R$.  
Let us denote by $A$ a subset of players described by a set $A = \{P_1,...,P_j \} \subseteq \mathcal{P} $. 
 Distribution of shares is given by a completely positive isometric map,
\begin{equation}
    \Lambda_D: S(\mathcal{H}_S) \rightarrow S(\mathcal{H}_1 \otimes ....... \otimes \mathcal{H}_m) 
\end{equation}
where $S(\mathcal{H}_A)$ represents the state space of the system $A$. We denote by $\ket{RP_1....P_m}$ the state of quantum system $RP_1....P_m$ after applying $\Lambda_D$ to S, and identity to $R$. Now, we provide the information-theoretic definition of a quantum secret sharing scheme based on \cite{r9}. 

\begin{definition}[\cite{r9}]
Let $R$ be a reference system such that $SR$ is in a pure state. A quantum secret sharing scheme is a completely positive map which distributes a quantum secret $S$  among a set of players $\mathcal{P} = \{P_1....P_m \}$  such that it realizes an access structure $\Gamma = (\Gamma_{\text{YES}},\Gamma_{\text{NO}})$ with the following two requirements:
\begin{itemize}
\item recoverability requirement: For all $A \in \Gamma_{\text{YES}}$, we have that $I(R:A) = I(R:S)$.
\item secrecy requirement: For all $A \in \Gamma_{\text{NO}} $, we have $I(R:A) = 0$
\end{itemize}
\end{definition}
Such a scheme is also called a threshold or "all or nothing" scheme because every set of players in $\Gamma_{\text{YES}}$ can reconstruct the secret perfectly, while a set in $A \in \Gamma_{\text{NO}} $ has no information about it. Furthermore, it is known that these conditions are necessary and sufficient for quantum error correction \cite{PhysRevA.54.2629}. A quantum threshold scheme is also denoted by $((k,n))$ where an arbitrary quantum secret is divided into $n$ parties such that any $k$ or more parties can perfectly reconstruct the secret while any $k-1$ or fewer parties have no information at all about the secret. This also means that the reduced density matrix of these $k-1$ parties is independent of the secret.

\begin{definition}[\cite{r8}]
Let $1 \leq k \leq n$. A $k$ out of $n$, $((k,n))$, quantum threshold access structure $\Gamma$ over a set of players $\mathcal{P} = \{P_1....P_n \}$ is the complete access structure accepting all subsets of size at least $k$, that is $\Gamma_{\text{YES}} = \{ A \subseteq \mathcal{P}: |A| \geq k \}$ and $\Gamma_{\text{NO}} = \{ A \subseteq \mathcal{P}: |A| \leq k -1  \}$
\end{definition}
 An example of a $((2,3))$ quantum threshold scheme is given in Appendix \ref{app:example}. In quantum secret sharing schemes, it is possible to encode the quantum secret into the players, which can be in a pure or mixed state. A pure state quantum secret sharing is a scheme to encode quantum secrets as a pure state, while a mixed state quantum secret sharing is a scheme to encode the quantum secret as mixed states \cite{r7, r8}. Furthermore, the scheme itself can be perfect or approximate. 

\begin{definition}
Let $R$ be a reference system such that $SR$ is in a pure state. A approximate quantum secret sharing is a scheme where a quantum secret $S$ among a set of players $\mathcal{P} = \{P_1....P_m \}$ such that it realizes an  access structure $\Gamma = (\Gamma_{\text{YES}},\Gamma_{\text{NO}})$ with relaxed conditions:
\begin{itemize}
    \item Relaxed recoverability requirement: For all $ A \in \Gamma_{\text{YES}}  $, we have $I(R:A) \geq I(R:S) - \delta$ for some small $\delta$. 
    \item Relaxed secrecy requirement: for all $A  \in \Gamma_{\text{NO}} $, we have $I(R:A) \leq \gamma$ for some small $\gamma$.   
\end{itemize}
\end{definition}

A relaxed recoverability requirement guarantees that the authorized parties can reconstruct the secret with high fidelity. The relaxed secrecy requirement ensures that the correlation between the quantum secret $S$ and unauthorized parties $B$ is small, preventing unauthorized parties from having much information about the secret. 

\begin{figure}
    \centering
    \includegraphics[width=0.5\textwidth]{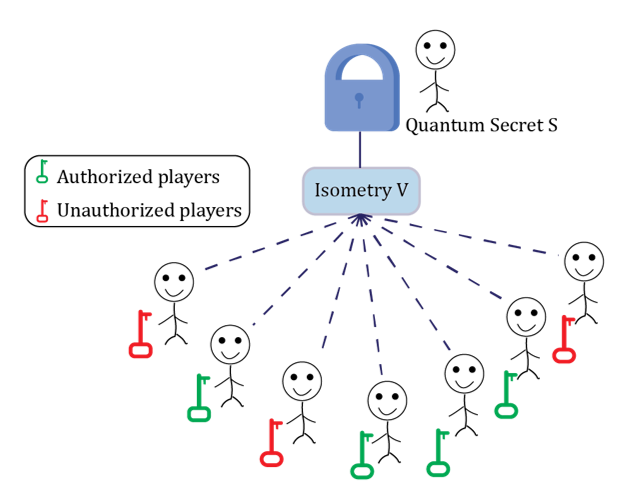}
    \caption{$((k,n))$ Quantum Secret Sharing(QSS) threshold schemes; $k$ players out of $n$ players have to come together to reconstruct the original secret. Any arbitrary collection of players larger than size $k$ is called authorized players, while one with fewer than $k$ is called unauthorized players.}
    \label{fig:my_label}
\end{figure}

Unlike in threshold secret sharing schemes, a ramp secret sharing protocol exists where partial information about the secret is allowed to leak to a set of participants that cannot fully reconstruct the secret. In particular, there are parameters $1 \leq b < g \leq n$ such that the set of parties of size at least $g$ should be able to reconstruct the secret,  while sets of parties of size at most $b$ cannot learn anything about the secret. Furthermore, we make no requirement on the parties of size $g-b$. $g$ and $b$ stand for good and bad, respectively. Such a scheme is denoted by $(b,g)$ ramp secret sharing scheme. A quantum version would be denoted by $((b,g))$. The classical ramp secret sharing scheme was first introduced by Blakley and Meadows \cite{ramp_scheme}, and later used to construct the efficient secure multiparty computation protocols \cite{ramp_scheme_2}. Compared to quantum threshold secret sharing, quantum ramp secret sharing has not been extensively explored \cite{Zhang_2014, MATSUMOTO_2018}.

\section{Our contributions}
In this section, we provide an overview of our main contributions, outlined as follows
\begin{itemize}
    \item In the literature, one can find various notions of quantum information scrambling. We have provided a novel information-theoretic definition of scrambling that encompasses these other notions. This information-theoretic definition, which we refer to as $l$ scrambling, goes as follows. Suppose $RA$ is in a pure state. A unitary operator $U_{AB}: A \otimes B \xrightarrow{} C \otimes D$ is called $l$-scrambling if any arbitrary subsystem $C$  of size less than $l$ approximately decouples from $R$.
    
    \item Information theoretic definition of scrambling allows us to show that Haar scrambling, a random unitary chosen from Haar measure, is equivalent to  $ \left( \left( \frac{N+s(\mathcal{P})}{2}- \epsilon, \frac{N+s(\mathcal{P})}{2}+\epsilon \right) \right)$ quantum ramp secret sharing scheme. Here $N$ is the total size of players $\mathcal{P}$, $s(\mathcal{P})$ is the Von Neumann entropy of the players, and $\epsilon$ is a measure of accuracy. 
    
    \item  The simple counting argument shows that it requires an exponential number of gates, $\mathcal{O}(4^n)$, to implement the $n$ qubit Haar unitary. This would make this protocol exponentially hard to implement. However, using $t$ design arguments, we show that the protocol, both encoding and decoding, can be implemented efficiently, i.e., using only a polynomial number, $\mathcal{O}(n^2)$, of one and two-qubit gates. 
\end{itemize}

\section{Quantum Information Scrambling}
\label{sec:quantumInformationScrambling}
The phenomenon of "scrambling" of quantum information is a fundamental concept closely related to many-body physics, quantum chaos, complexity theory, and black holes \cite{r10,r11}. This phenomenon can be viewed as a form of "thermalization of the quantum information" and is crucial for understanding the spread of information in quantum systems.
As illustrated in Figure \ref{fig:informationScrambling2}, at time $t=0$, the reference $R$ is only entangled with the Alice qubit $q_A$, such that $S(R) = S(q_A) = 1$, $S(Rq_A) = 0$, and $I(R:q_A) = 2$. As time progresses, the quantum information from Alice's qubit becomes increasingly scrambled, depending on the system's dynamics. This is because the initially localized entanglement of Alice's qubit with the reference $R$ begins to expand across the system. In one scenario, the dynamics can be such that the information is carried coherently, for example, through a circuit using only SWAP gates. The information is highly localized in this case, and the unitary cannot be considered scrambling. Conversely, in strongly interacting systems, the entanglement spreads in a complex fashion such that the reference $R$ becomes entangled with a large collection of qubits. In this case, the information is considered to be scrambled.
Because quantum mechanics is unitary, information cannot be lost globally. Therefore, to define information scrambling, we must focus on subsystems.

To study the spreading of information in quantum dynamics, we consider the following simple model where we have a Hilbert space of $n = |A| + |B| = |C| + |D|$ qudits partitioned as:
\begin{equation}
    \mathcal{H} = A \otimes B = C \otimes D
\end{equation}
and a unitary map
\begin{equation}
    U_{AB}: A \otimes B \xrightarrow{} C \otimes D
\end{equation}
Suppose Alice chooses to encode her quantum information in system $A$. Additionally, we introduce an external reference system $R$, which is perfectly entangled with $A$. The Alice system then interacts with another system $B$, which is purified by the external system $B'$.   Thus, the total initial state is $\ket{RA} \ket{BB'}$, and the final state is:
\begin{equation}
\ket{\psi}_{RB'CD} = U_{AB} \otimes I_{RB'}\ket{RA} \ket{BB'}
\end{equation}
The overall setup is depicted in Figure \ref{fig:genralCircuit}. In the literature, various notions of quantum information scrambling are presented, which we have classified into three levels of hierarchy. The first level is Haar scrambling, where the unitary $U$ is a random quantum circuit selected from the Haar measure. The objective is to determine whether a part of the system is maximally mixed and, if not, how close it is to being in a maximally mixed state. A maximally mixed state represents a state in which information is thermalized. At the second level, we encounter chaotic scrambling, where the Hamiltonian governing the quantum dynamics exhibits quantum chaotic behavior. Chaotic scrambling includes Haar scrambling, as proven in \cite{yoshida2017efficient}, but can also involve other families of unitaries. At the third level, we provide an information-theoretic definition of scrambling, encompassing both Haar and chaotic scrambling. 
 
 \begin{figure}
    \centering
    \includegraphics[scale=0.8]{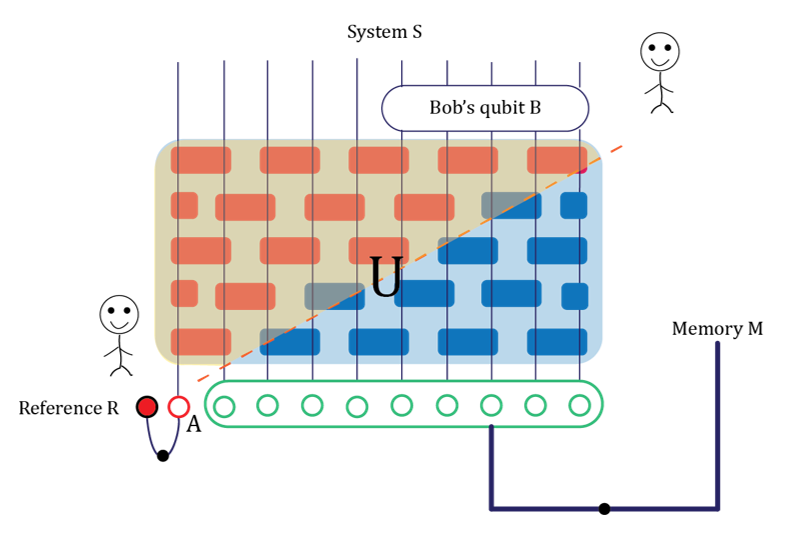}
    \caption{Alice and Bob try to communicate via a scrambling unitary $U$ acting on N qubits. Alice has full control of the first qubit $A$, which is entangled with the reference system $R$. The rest of the qubits, marked green, can be in a mixed state  $\rho$ purified by an external system called Memory. This pure state is denoted by $\ket{\sqrt{\rho}}$.  The scrambling unitary acts only on Alice's qubit and the system's qubit, which is marked green. After the unitary evolution, Bob then has access to the system's set of qubits. If the scrambling unitary has a local structure, it is also possible to have an information light cone indicated by a yellow cone. The initial state of the system, including the reference, the system, and the memory, is:
     $\ket{\boldsymbol{\psi}} = \frac{1}{\sqrt{2}}(\ket{0}_R\ket{0}_{q_A} +\ket{1}_R\ket{1}_{q_A}) \ket{\sqrt{\rho}}$
Since the time evolution unitary doesn't act on reference and memory, the total dynamics is then given by:
$\ket{\boldsymbol{\psi}(t)}  = \frac{1}{\sqrt{2}} \left(\ket{0}_R U_S \otimes I_\text{MEM} \ket{0,\sqrt{\rho}}  +  \ket{1}_R U_S \otimes I_\text{MEM} \ket{1,\sqrt{\rho}} \right) $. More detail about this figure is in Appendix \ref{app:toymodel}.  }
    \label{fig:informationScrambling2}
\end{figure}

\subsection{Haar scrambling}
In Haar scrambling, we have a unitary $U$ that is chosen at random from the Haar measure. This serves as a baseline for scrambling, as we expect the late-time values of entropy and mutual information for any scrambling unitary $U(t)$ to match those of Haar random unitaries. However, a pure random unitary circuit of size $2^n \times 2^n$ does not have any inherent local structure. To account for locality, one can restrict the situation to be local in one or two-qubit gates.

\begin{definition}[\cite{r11}]
A Haar scrambled quantum state is the state obtained by applying a random unitary $U$ chosen from group-invariant Haar measure, $\ket{\psi (U)} = U \ket{\psi_0}$.
\end{definition}

The Haar measure possesses several interesting mathematical properties, enabling the derivation of analytical results for various information-theoretic quantities \cite{mele2023introduction}. We briefly review some of these properties in Appendix \ref{app:HaarIntegrals}. One important physical characteristic of Haar scrambling is that a randomly chosen pure state in $\mathcal{H}_{AB}$ is likely close to a maximally entangled state if $\frac{|A|}{|B|} << 1$. This property is known as Page scrambling, named after Page, who first introduced it to explain black hole physics \cite{Page_1993}. Mathematically, for any bipartite Hilbert space $\mathcal{H}_A \otimes \mathcal{H}_B$, we have: 

\begin{equation}
    \int dU || \rho_A(U) - \frac{I_A}{|A|} ||_1 \leq \sqrt{\frac{|A|^2 - 1 }{|A||B| +1 }} < \frac{\rho + 2}{1}
\end{equation}
and when $|B|$ is significantly larger than $|A|$, the typical deviation of $\rho_A$ from maximally mixed state is extremely small. For example, if $B$ has 10 more qubits than $A$, then the typical deviation from maximal entanglement is bounded by $2^{-5}$. Using the analytical results from Appendix \ref{app:HaarIntegrals}, it is possible to analytically compute the R\'enyi entropy between the reference $R$ and the subsystem $C$ using Eq. \ref{eq:haarintegral}. Unfortunately, the quantity $I_2(X:Y) = S_2(X) + S_2(Y) - S_2(XY)$ is not a good candidate for mutual information as it violates the data processing inequality \cite{Linden_2013}. Moreover, it is generally not even possible to use this quantity to obtain bounds for the quantum mutual information, which is the actual quantity of interest. However, for the Haar scrambling case, $R$ is in a maximally mixed state, and $C$ is expected to be close to a maximally mixed state \cite{Hayden_2006}. In this case, we can use R\'enyi entropy to upper-bound the mutual information as follows:
\begin{equation}
\label{eq:R'enyiBound}
I(R:C) \leq \log|R| + \log|C| - S_2(RC),
\end{equation}
where $S_2$ denotes the second R\'enyi entropy. For example, let us consider a case where $A$ is a single qubit. Let $p$ be the size of region $C_p$.  The R\'enyi entropy of $C_p$ and $RC_p$ can be analytically computed using Eq. \ref{eq:haarintegral} and can be found in \cite{r12, r10, r13}. Then, using Eq. \ref{eq:R'enyiBound}, we can upper-bound the quantum mutual information between $R$ and $C_p$ as
\begin{equation}
\label{eq:R'enyiMutualInformation}
I(R:C_p) \leq 1 + \log_2 \left( 2 - \frac{3(1 - 2 ^{2p -2 N})}{2 + (2^{-s} - 2^{-N + 1})4^{p- N/2}} \right),
\end{equation}
where $s$ is the second R\'enyi entropy of system $B$.

\subsection{Chaotic Scrambling (OTOC measure)}
In this section, we review the concept of chaotic scrambling, which involves a Hamiltonian exhibiting quantum chaotic features. Quantum chaotic Hamiltonians are typically strongly interacting systems that spread quantum information throughout the whole system.
The measurement of Out-of-Time-Ordered Correlators (OTOCs) \cite{larkin, Maldacena_2016, Hashimoto_2017} is often employed to gauge quantum chaos. For a brief overview of OTOCs, refer to Appendix \ref{app:OTOCs}. To formulate the chaotic scrambling
 in the language of quantum channels, we consider a Hilbert space of $n = |A| + |B| = |C| + |D|$ qubits that are partitioned as:
$\mathcal{H} = A \otimes B = C \otimes D$, and a unitary map $U_{AB}: A \otimes B \xrightarrow{} C \otimes D$. Let the dimension for the system be denoted by $d$. Suppose $a$, $b$, $c$, and $d$ corresponds to the size of $A$, $B$, $C$ and $D$ respectively and the corresponding dimensions be denoted by $d_A$, $d_B$, $d_C$ and $d_D$ respectively.  For systems consisting of qubits only, it would be $d= d_Ad_B = 2^a 2^b= d_Cd_D = 2^c2^d = 2^{a+b} = 2^{c+d} = 2^n$.
Let $O_i$ be hermitian operators supported on region $i$, where $i \in { A,B,C,D}$ and define $O_i(t) = U^\dagger O(0) U$. With this setup, we refer to the chaotic scrambling as unitary exhibiting chaotic features given by the OTOCs measure \cite{yoshida2017efficient}.

\begin{definition}[\cite{yoshida2017efficient}]
A unitary operator $U_{AB}$ is called chaotic scrambling iff it satisfies:
\begin{equation}
\label{eq:chaoticscrambling}
  \langle O_A(0) O_B(t) O_C(0) O_D(t) \rangle\simeq \langle O_AO_C \rangle \langle O_B \rangle \langle O_D \rangle + \langle O_BO_D \rangle \langle O_A \rangle \langle O_C \rangle - \langle O_A \rangle \langle O_B \rangle \langle O_C \rangle \langle O_D \rangle  
\end{equation}
 where 	$\simeq$ means up to an order $d^{-2}$. 
\end{definition}
In terms of the size of subsystems, this definition implies that a unitary $U_{AB}$ is  chaotic scrambling if it satisfies:
\begin{equation}
\langle O_A(0) O_B(t) O_C(0) O_D(t) \rangle \simeq \frac{1}{d_A^2} + \frac{1}{d_D^2} - \frac{1}{d_A^2 d_D^2}
\end{equation}
However, it's important to note that defining quantum chaos via OTOCs doesn't always correspond to true quantum chaos \cite{Xu_2020}. When averaging $\langle O_A(0) O_B(t) O_C(0) O_D(t) \rangle $ over Haar random unitary $U$, it also satisfies the condition for chaotic scrambling. Thus, Haar scrambling belongs to the family of chaotic scrambling.

\subsection{$l$ scrambling}
\label{sec:informationScrambling}
\begin{figure}
    \centering
    \includegraphics[width=0.6 \textwidth]{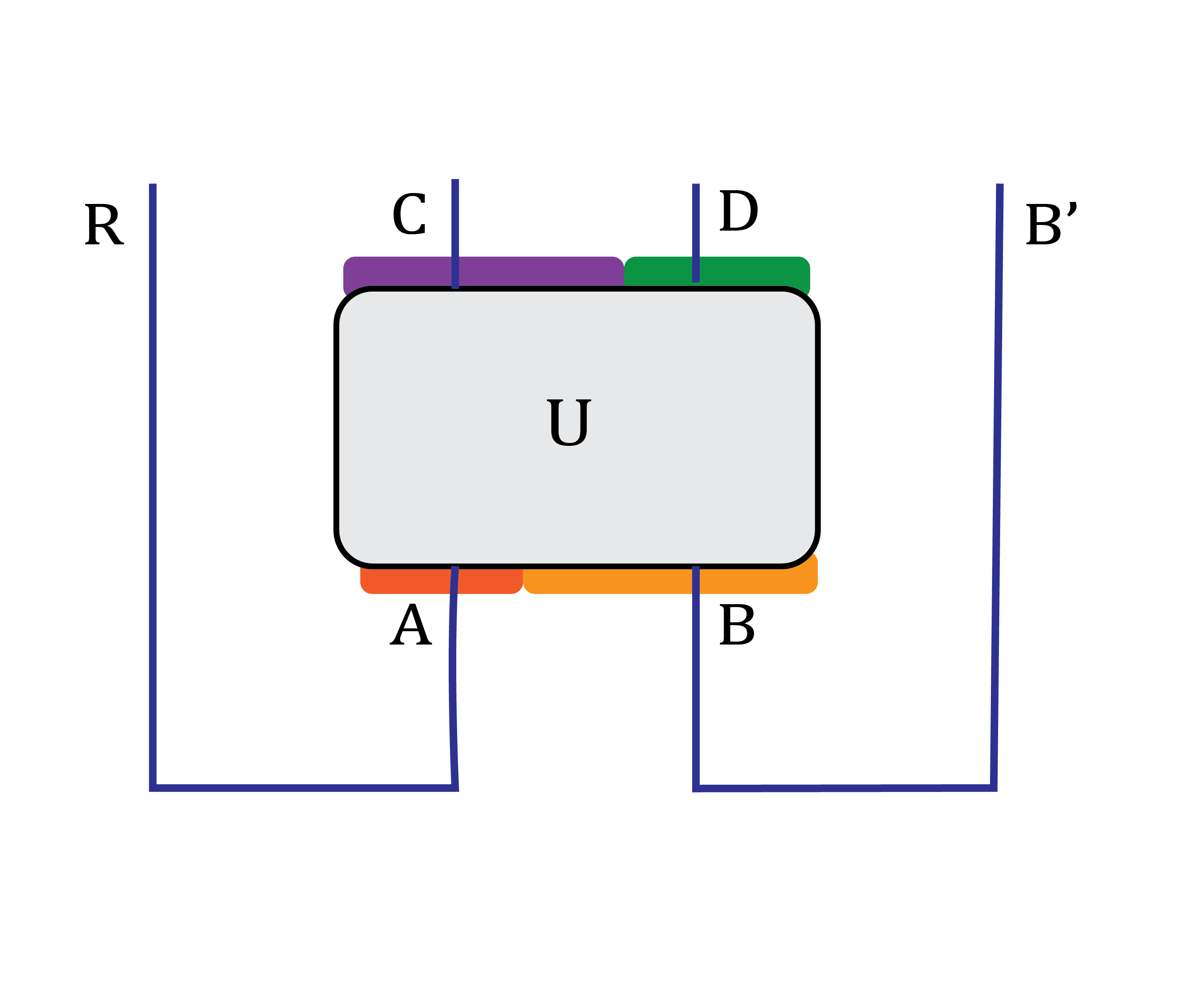}
    \caption{ Suppose Alice decides to encode her quantum information in $A$, which is perfectly entangled with external reference state $R$. Alice's qubit then interacts with system $B$, which is also purified with external system $B'$. Thus the total initial state is $\ket{RA} \ket{BB'}$ while final state after scrambling unitary $U_{AB}$ is
$    \ket{\psi}_{RB'CD} = U_{AB} \otimes I_{RB'}\ket{RA} \ket{BB'}$.}
    \label{fig:genralCircuit}
\end{figure}

During scrambling, a generation of global entanglement causes the initially localized quantum information to spread across the system, thereby hiding it from the observer, who can only access a portion of it. Consequently, information-theoretic measures like entanglement entropy and mutual information emerge as natural and powerful tools for formally defining quantum information scrambling. To give an information theoretic setting, we will again consider the Hilbert space of $n = |A| + |B| = |C| + |D|$ qubits partitioned as:
\begin{equation}
    \mathcal{H} = A \otimes B = C \otimes D
\end{equation}
and a unitary map
\begin{equation}
    U_{AB}: A \otimes B \xrightarrow{} C \otimes D
\end{equation}
Suppose Alice chooses to encode her quantum information in system $A$. Additionally, we introduce an external reference system $R$, which is perfectly entangled with $A$. The Alice system then interacts with another system $B$, which is purified by the external system $B'$.   Thus, the total initial state is $\ket{RA} \ket{BB'}$, and the final state is:
\begin{equation}
\ket{\psi}_{RB'CD} = U_{AB} \otimes I_{RB'}\ket{RA} \ket{BB'}
\end{equation}
The overall setup is depicted in Figure \ref{fig:genralCircuit}.

Thus, it becomes evident that a definition of scrambling must entail that accessing a small part of the subsystem $C$ should not grant one the ability to obtain much information about system $A$. In the literature, it is frequently assumed that the state $BB'$ is maximally entangled \cite{r11}. In such a scenario, system $C$ becomes decoupled from $R$ provided that the size  $D$ is larger than that of $A$. However, it becomes more complicated when $BB'$ does not form a maximally entangled state. To address this issue, we propose a new definition of scrambling unitary where $R$ and $C$ are nearly independent for most $C$ with a size smaller than some parameter $l$, which interestingly depends upon the purity of $B$. 

\begin{definition}[$l$- scrambling]
\label{def_l-scrambling}
Suppose $RA$ is in a pure state. A unitary operator $U_{AB}: A \otimes B \xrightarrow{} C \otimes D$ is called $l$-scrambling if any arbitrary subsystem $C$  of size less than $l$ approximately decouples from $R$: $\rho_{RC} \stackrel{\epsilon}{\approx} \rho_R \otimes \rho_C$ that is,
\begin{equation}
   F(\rho_{RC},\rho_R \otimes \rho_C ) \geq 1 - \epsilon \text{ for some } \epsilon \geq 0 \text{ for any } C \text{ with } |C| < l 
\end{equation}
\end{definition}
Hence, the outcome of any measurement on $C$ is statistically independent of any measurement on $R$. Moreover, the decoupling theorem \cite{decoupling, Hayden_2006} implies that having access to the system $DB'$ allows for the existence of a decoder that can reconstruct a state $\rho_A$ with high fidelity. Thus, $l$-scrambling is related to both information hiding and retrieval. 

The parameter $l$ depends on the purity of the state $B$ but is independent of how information was scrambled, i.e., on the structure of the scrambling unitary $U_{AB}$.
However, there are other quantities, such as scrambling time (the time taken to scramble quantum information) and the complexity of the scrambling unitary, which depends on the structure of the unitary $U_{AB}$. These quantities are crucial while studying the dynamics of physical systems and quantum circuits. It is worth noting that this definition of $l$-scrambling doesn't allow us to make a statement for system $C$ with sizes greater than $l$. For example, subsystem $D$ could have partial or no correlation with $R$.


Another measure of information scrambling often considered in the literature is the tripartite information \cite{r11, Kitaev_2006}, denoted as $I_3(R:C:D) = I(R:C) + I(R:D) - I(R:CD)$. This quantifies the information about system $A$ that is non-locally hidden over arbitrary subsets $C$ and $D$. 
It is commonly asserted that a quantum system that scrambles information must have a non-positive tripartite information $I_3(A:C:D)$, and that the maximal negativity of $I_3(R:C:D)$ is indicative of the circuit that scrambles information most effectively. However, this notion only applies when $B$ is maximally mixed. In fact, for any unitary $U_{AB}: A \otimes B \xrightarrow{} C \otimes D$, we have $I(R:C) + I(R:D) \leq I(R:CD)$ and $I_3(R:C:D)$ is always non-positive. Hence, it cannot serve as a reliable indicator of scrambling. To overcome this issue, we propose a new variant of tripartite information based on the definition of $l$-scrambling, $l$-tripartite information, denoted by $I_3(R:C_l:D_l)$. Here, $C_l$ and $D_l$, which may or may not overlap, have sizes less than $l$, but $C_lD_l$ together have a size greater than or equal to $l$.

\begin{definition}[$l$- tripartite information]
Suppose we have a quantum state $RCD$ which can be pure or mixed. We define l-tri partite information, $I_3(R:C_l:D_l)$ as:
\begin{equation}
    I_3(R:C_l:D_l) = I(R:C_l) + I(R:D_l) - I(R:C_lD_l)
\end{equation}
where $C_l$ and $D_l$ are subsystems of $CD$ with $|C_l| < l$,  $|D_l| < l$ and $|C_lD_l| \geq l$. 
\end{definition}

\begin{remark}
Suppose $RA$ is in a pure state. A unitary operator $U_{AB}: A \otimes B \xrightarrow{} C \otimes D$ is $l$-scrambling iff $I_3(R:C_l:D_l)$  is maximally negative i.e. $I_3(R:C_l:D_l) = - I(R:C_lD_l) + \mathcal{O}(\epsilon)$ for some $\epsilon \geq 0$.
\begin{proof}
If $U_{AB}$ is $l$- scrambling, then system $C_l$ and $D_l$ decouples from $R$, thus  $I(R:C_l) = I(R:D_l) = \mathcal{O}(\epsilon)$. This implies $I_3(R:C_l:D_l) = - I(R:C_lD_l) + \mathcal{O}(\epsilon)$. Furthermore, if  $I_3(R:C_l:D_l) = - I(R:C_lD_l) + \mathcal{O}(\epsilon)$, then $I(R:C_l)$ and $I(R:D_l)$ are of order $\mathcal{O}(\epsilon)$. Thus, $C_l$ and $D_l$ decouples from $R$ which implies that $U_{AB}$ is $l$- scrambling. 
\end{proof}
\end{remark}

Although there are three distinct notions of quantum information scrambling, a hierarchy exists in that $l$-scrambling encompasses both chaotic and Haar scrambling. Chaotic scrambling includes Haar scrambling unitaries \cite{r11, yoshida2017efficient, Harrow_2021, Hayden_2006}. However, the reverse is not always true. $l$-scrambling can be constructed using only Clifford gates, which lack quantum chaotic features \cite{leone2023learning, oliviero2022black}. And chaotic unitaries may include other dynamics, such as the mixed-field Ising model, which are not Haar unitary. Since calculating OTOCs typically involves averaging over maximally mixed states, $l$-scrambling offers a more nuanced understanding of information scrambling by considering purity. We plan to establish a more detailed relationship between $l$-scrambling and OTOCs in the future.


\section{Ramp secret sharing from Haar scrambling}
\label{sec:RampSecret}

 In the definition of $l$-scrambling, we saw that the subsystem with a size less than $l$ has little correlation with the reference system. In secret sharing terminology, this subsystem would be referred to as the unauthorized set. In this section, we will formally show that Haar scrambling exhibits characteristics of ramp secret sharing. This scheme, rather than a threshold one, is expected, as the threshold scheme is hard to obtain from a randomly selected unitary operation. To do the proper analysis, we will first provide an information-theoretic definition of the quantum ramp secret sharing scheme. 

\begin{definition}
Let $1 \leq b < g \leq n$, and $R$ be a reference system such that $SR$ is in a pure state. The $((b,g))$ quantum ramp secret sharing is a a completely positive map which distributes a quantum secret $S$ among a set of players $\mathcal{P} = \{P_1....P_m \}$ such that it realizes an incomplete access structure $\Gamma_{b,g} = (\Gamma_{\text{YES}},\Gamma_{\text{NO}})$ with conditions:
\begin{itemize}
    \item Recoverability requirement: If $I(R:A) = I(R:S)$ for all $|A| \geq g$ then $A \in \Gamma_{\text{YES}}$. This implies that $A$ can reconstruct the secret.
    \item Secrecy requirement: If $I(R:A) = 0$ for all $|A| \leq b$, then $A \in \Gamma_{\text{NO}}$. This implies that A cannot learn anything about secret S. 
    \item No requirements on sets A where $b < |A| < g$ which we refer as grey area. 
\end{itemize}
\end{definition}

The scheme is then referred to as the $((b,g))$ ramp scheme and can be made approximate with relaxed recoverability and relaxed secrecy requirements as follows. 

\begin{definition}
Let $1 \leq b < g \leq n$, and $R$ be a reference system such that $SR$ is in a pure state. The $((b,g))$ approximate quantum ramp secret sharing is a a completely positive map which distributes a quantum secret $S$ among a set of players $\mathcal{P} = \{P_1....P_m \}$ such that it realizes an incomplete access structure $\Gamma_{b,g} = (\Gamma_{\text{YES}},\Gamma_{\text{NO}})$ with conditions:
\begin{itemize}
    \item Relaxed recoverability requirement: If $I(R:A) \geq I(R:S) - \delta$ for all $|A| \geq g$ then $A \in \Gamma_{\text{YES}}$.
    \item Relaxed secrecy requirement: If $I(R:A) \leq \gamma$ for all $|A| \leq b$, then $A \in \Gamma_{\text{NO}}$.
    \item  No requirements on sets A where $b < |A| < g$
\end{itemize}
  
\end{definition}
 The relaxing parameters $\delta$ and $\gamma$ give a measure of how well authorized parties can reconstruct the secret while unauthorized parties are prevented from recovering it. It the the scheme of this nature that Haar scrambling shows rather than a perfect one. Figure \ref{fig:approximateRamp} shows the relation between mutual information and the size of parties for perfect and approximate threshold and ramp schemes when a single qubit secret is shared among $N$ parties.
 
 To analyze the quantum ramp schemes $((g,b))$, we introduce two metrics: the gap $G$ and the Rampiness $R$. The gap is defined as the difference between $g$ and $b$, expressed as $G = g - b$. The Rampiness, denoted by $R$, is the ratio of the gap to the total number of players, represented as $R = \frac{G}{N}$.  The Rampiness value ranges from zero to one, where $R = 0$ corresponds to the threshold scheme with a gap $G = 0$, and $R = 1$ signifies the absence of any secret sharing scheme with a gap $G = N$. A lower value of $R$ suggests a superior scheme since it indicates a smaller grey region. Naturally, we can extend this definition of rampiness and gap to the approximate scenario by selecting $g$ and $b$ based on relaxed parameters $\delta$ and $\epsilon$.
 
\begin{figure}
    \centering
    \includegraphics[width=0.7\textwidth]{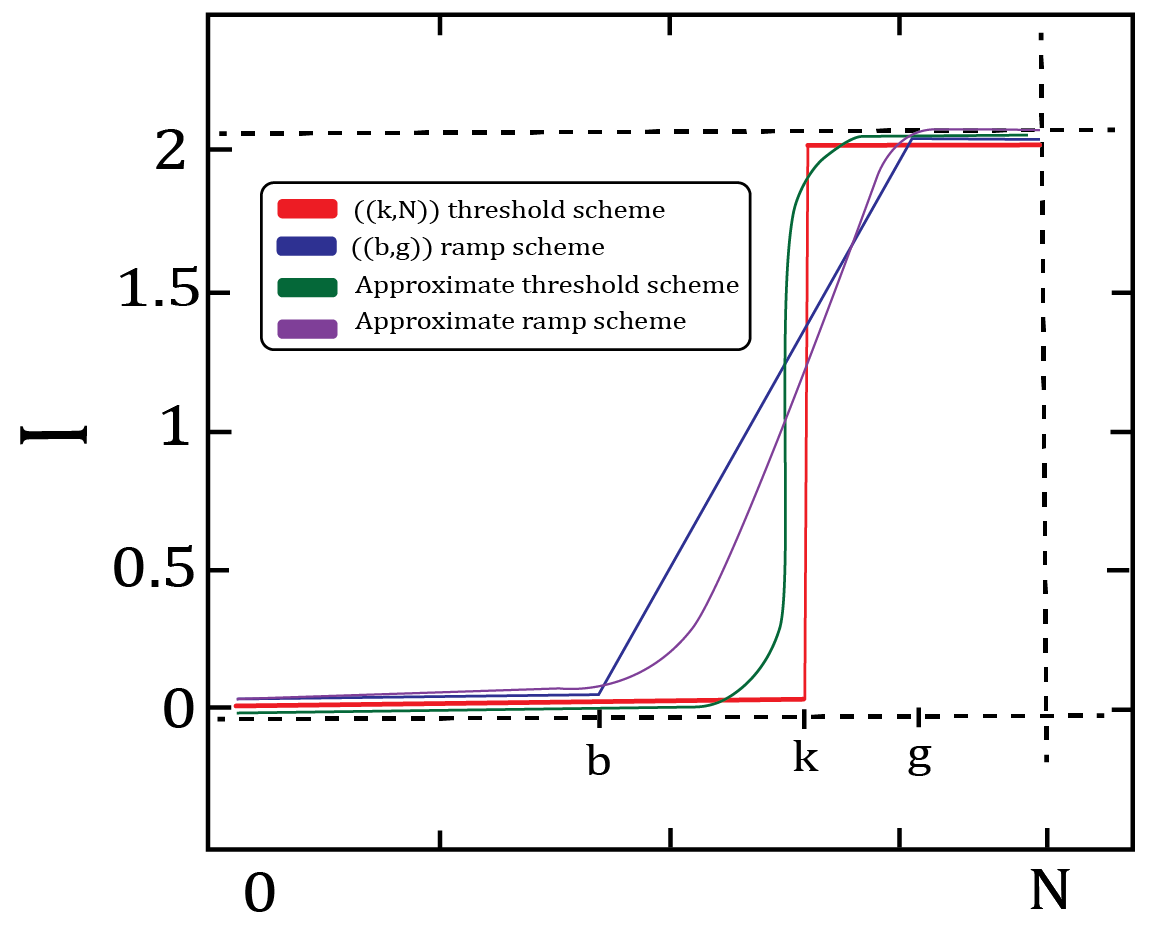}
   \caption{A qubit secret shared among $N$ parties in different scenarios. For the $((k,N))$ scheme, the mutual information between the reference $R$ and the parties goes through a sharp step function-like transition from $0$ to $2$ as the number of parties crosses the threshold $k$. In $((b,g))$ ramp scheme, the mutual information is zero for parties less than the size $b$ and two when it is greater than $g$ while some mutual information in the grey area i.e. when the player's size is greater than $b$ but less than $g$. Fig b:   Includes also approximate $((k,N))$ scheme as for all $|A| \geq k$, $I(R:A) \geq 2 - \delta$ for some small $\delta$ and for all $|A| \geq k$, $I(R:A) \leq \epsilon$ for some small $\epsilon$ and the approximate $((b,g))$ scheme as $I(R:A) \geq 2 - \delta$ for all $|A| \geq g$ and $I(R:A) \leq \epsilon$ for all $|A| \leq b$. }
    \label{fig:approximateRamp}
\end{figure}

\begin{remark}
The lower bound on the size of gap, $G = g -b$ is equal to the size of the quantum secret.
\begin{proof}
We call $S$ an important share if there is an unauthorized set T such that $S \cup T $ becomes authorized.
The dimension of each important share of a quantum secret sharing scheme must be at least as large as the dimension of the secret \cite{r8}. $G$ is an important share, as adding size $G$ parties to the size $b$  parties gives us an authorized set. Thus, $G$ must at least have the size of a quantum secret. 
\end{proof}
\end{remark}

Having established the formal definition of quantum ramp secret sharing, let us delve into the relationship between secret sharing and Haar scrambling. The parameter $l$ of $l$-scrambling has to be less than or equal to the parameter $b$ of the quantum ramp secret sharing scheme: $l \leq b$. We discussed before that the parameter $l$ of $l$-scrambling depends upon the purity of the initial state. Therefore, we will see that, by changing the purity of players, all possible approximate quantum ramp secret-sharing schemes can be obtained.
First, we consider the scenario of a pure state and subsequently demonstrate how one can get the mixed state case by discarding players from the pure state case.

For this, let us denote the total system size, i.e., number of players $\mathcal{P} = \{P_1....P_m \}$ by $N$ and the entropy of the $\mathcal{P}$ before the quantum secret is encoded by $s(\mathcal{P})$. Let $1 \leq b < g \leq n$, $S$ be the quantum secret, and  $R$ be a reference system such that $SR$ is in a pure state, $((b,g))$ be a quantum ramp schem,e and $P(x)$ be the arbitrary set of players with size $x$.

\subsection{Ramp secret sharing from Haar scrambling: Pure state case}

A pure-state quantum secret sharing scheme encodes quantum secrets into pure states of the players.  
For a pure initial state of the players, the entropy satisfies $s(\mathcal{P}) = 0$, and for any two disjoint subsets $P(b)$ and $P(N-b)$ we have
\begin{equation}
    I(R : P(b)) + I(R : P(N-b)) = I(R : S),
\end{equation}
since the global state on $RP(N)$ is pure.  
If $I(R : P(b)) = 0$, then $P(b)$ is an unauthorized set.  
Because the complementary region $P(N-b)$ cannot both be unauthorized (by the above identity) and cannot overlap with another authorized region (by the no-cloning theorem), one obtains the constraints
\begin{align}
    \frac{N+1}{2} + 1 \;\le\; g \;\le\; N - b, \\
    N - g \;\le\; b \;\le\; \frac{N-1}{2}.
\end{align}

We now analyze the case where the scrambling unitary $U$ is drawn from the Haar measure.  
Let $P(\ell)$ be any subsystem of size $\ell$ after scrambling, and let $I(P(\ell))$ denote the mutual information between the reference $R$ (one qubit) and this region.  
Using the average second R\'enyi entropy of Haar-random states (Appendix~B) together with the bound
\begin{equation}
    I(R:C) \le \log|R| + \log|C| - S_2(RC),
\end{equation}
we obtain the general inequality
\begin{equation}
    I(P(\ell)) \;\le\; 2\ell - N + \log_2\!\left( 1 + 2^{\,N - 2\ell} \right).
    \label{eq:HaarPureBound}
\end{equation}

This expression shows that Haar scrambling cannot realize a threshold secret-sharing scheme:
for $\ell = N/2$, one finds $I(P(\ell)) \le 1$, strictly below the value $2$ required for perfect recovery of a qubit secret.  
Thus $\ell = N/2$ lies inside the ``grey region'' of a ramp scheme.  
We now show that Haar scrambling realizes an approximate ramp secret-sharing scheme.

\begin{theorem}
\label{theorem:purestate}
A unitary $U$ drawn from the Haar measure implements an approximate 
$\big( (\frac{N}{2}-\epsilon,\; \frac{N}{2}+\epsilon) \big)$
quantum ramp secret sharing scheme when the quantum secret is encoded into pure-state players.
\end{theorem}

\begin{proof}
Set $\ell = \frac{N}{2} - \epsilon$ in Eq.~\eqref{eq:HaarPureBound}.  
The linear term becomes
\[
2\ell - N = -2\epsilon,
\]
and the exponent in the logarithmic term becomes
\[
N - 2\ell = 2\epsilon.
\]
Thus,
\begin{equation}
    I(P(\ell))
    \;\le\;
    -2\epsilon
    \;+\;
    \log_2\!\big(1 + 2^{\,2\epsilon}\big).
    \label{eq:IPellIntermediate}
\end{equation}
Using $\log_2(1+x) \le x/\ln 2$,
\begin{equation}
    I(P(\ell))
    \;\le\;
    -2\epsilon
    \;+\;
    \frac{2^{2\epsilon}}{\ln 2}.
    \label{eq:IPellBound}
\end{equation}
For any target accuracy $\gamma > 0$, choose $\epsilon = O(\gamma)$ so that RHS~\eqref{eq:IPellBound} is at most $\gamma$.  
Therefore,
\[
    I(R : P(b)) \;=\; I\!\left(P\!\left(\tfrac{N}{2}-\epsilon\right)\right) \;\le\; \gamma,
\]
establishing the relaxed secrecy condition.

Since the global state on $RP(N)$ is pure, for all $\ell$ we have the identity
\begin{equation}
    I(P(\ell)) + I(R : P(N-\ell)) = I(R:S) = 2.
\end{equation}
Hence,
\[
I(R : P(N-\ell)) \;\ge\; 2 - \gamma,
\]
which is the relaxed recoverability requirement for authorized sets.  
Thus the authorized and unauthorized regions are characterized by
\begin{equation}
    b = \frac{N}{2} - \epsilon,
    \qquad
    g = \frac{N}{2} + \epsilon,
\end{equation}
completing the proof.
\end{proof}

The relaxed secrecy condition guarantees that any region $P(b)$ of size at most $b$ is nearly uncorrelated with the reference system.  
Using quantum Pinsker's inequality,
\begin{equation}
    \frac{1}{2\ln 2}
    \bigl(
        \Delta(\rho_{RP(b)}, \rho_{R}\!\otimes\!\rho_{P(b)})
    \bigr)^2
    \;\le\;
    I(R:P(b))
    \;\le\;
    \gamma,
\end{equation}
we obtain
\begin{equation}
\Delta(\rho_{RP(b)}, \rho_{R}\!\otimes\!\rho_{P(b)})
    \;\le\;
    O(\sqrt{\gamma}),
\end{equation}
showing that unauthorized parties learn essentially nothing about the secret.

Finally, for the pure Haar-scrambled case, the scheme gap is
\[
    G = g - b = 2\epsilon,
\]
which is independent of $N$, and the rampiness is
\[
    R = \frac{G}{N} = \frac{2\epsilon}{N}.
\]
Thus the rampiness decays polynomially in $N$, and the scheme becomes increasingly sharp as the system size grows. In Figure \ref{fig:haarPure1}, we have plotted the scheme when Haar scrambling a single qubit into a system of size 12, and in Figure \ref{fig:haarPure2}, into a system of other different sizes. One interesting observation is that while we increase the system size, the gap remains constant and is independent of the system size. This implies that the rampiness drops polynomially, as shown in figure \ref{fig:rampiness}. Thus, the secret sharing protocol improves as the system size increases.

\begin{figure}
    \centering
    \includegraphics[width=0.7\textwidth]{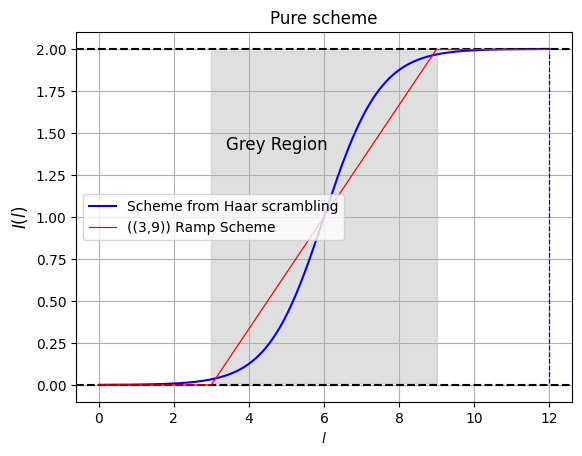}
    \caption{Mutual information between quantum secret and subsystem/parties of size $l$ for Haar scrambling and what we expect from the ramp scheme when the parties are in a pure state. Haar scrambling approximately behaves like a quantum ramp secret sharing scheme.  In figure \ref{fig:haarPure1}, we have plotted the mutual information between the reference qubit $R$ and $l$ arbitrary qubits in the system of size $12$ after applying Haar scrambling. The curve for Haar scrambling approximates the $((3,9))$ ramp scheme with a grey area when $l$ is between $3$ and $9$, which gives the gap $G = 6$ and rampiness $R= 0.5$.}
    \label{fig:haarPure1}
\end{figure}

\begin{figure}
    \centering
   \begin{subfigure}[b]{0.8\textwidth}
   \centering
         \includegraphics[width=\textwidth]{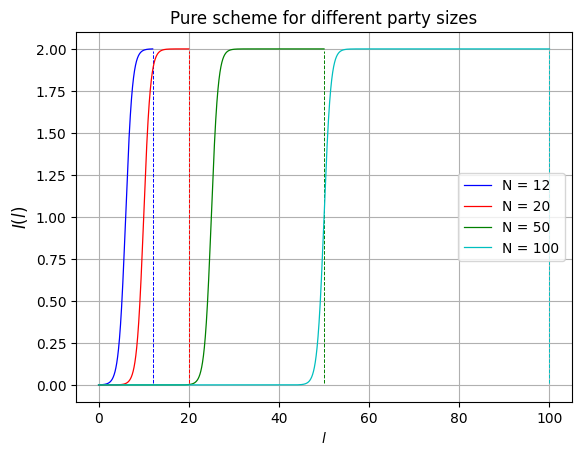}
         \caption{Pure schemes with different party sizes}
         \label{fig:haarPure2}
   \end{subfigure}
   \hfill
    \begin{subfigure}[b]{0.8\textwidth}
   \centering
         \includegraphics[width=\textwidth]{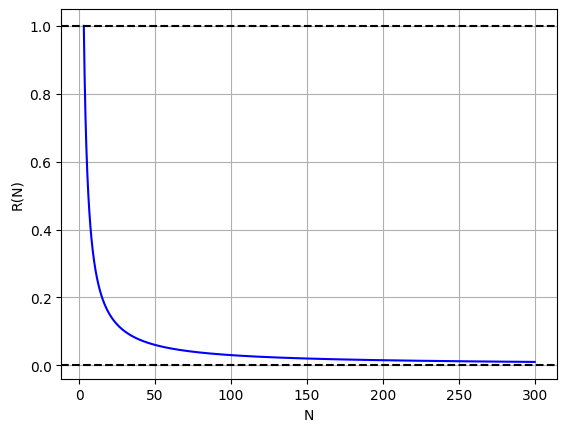}
         \caption{Rampiness with increasing party size $N$}
         \label{fig:rampiness}
   \end{subfigure}
   \caption{a: Quantum ramp secret sharing schemes from Haar scrambling for different party sizes for the pure case. As the party sizes increase, the gap becomes small compared to the total size. b: Rampiness drops down as the system/players' size increases. 
   This implies that secret sharing schemes become more effective for larger systems. }
    \label{fig:quantumSecretSharingb}
\end{figure}

\subsection{Ramp secret sharing from Haar scrambling: Mixed state case}

We now generalize the pure-state analysis to the case where the players
$\mathcal{P} = \{P_1,\ldots,P_m\}$ begin in a mixed state with nonzero
entropy $s(\mathcal{P})$.  
Let $N = |\mathcal{P}|$ denote the number of players, and let
$s(\mathcal{P})$ denote the von Neumann entropy of their joint state before
scrambling.  
For pure states ($s(\mathcal{P}) = 0$), Theorem~\ref{theorem:purestate}
shows that Haar scrambling implements an approximate ramp secret-sharing
scheme with parameters
\[
b_{\rm pure} = \frac{N}{2} - \epsilon,
\qquad
g_{\rm pure} = \frac{N}{2} + \epsilon.
\]
In the mixed case, we will see that the ramp parameters are shifted by the
initial entropy $s(\mathcal{P})$, yielding the more general expressions
\begin{equation}
\begin{aligned}
       g &= \frac{N + s(\mathcal{P})}{2} + \epsilon, \\
       b &= \frac{N + s(\mathcal{P})}{2} - \epsilon,
\end{aligned}
\end{equation}
which reduce to the pure-state values when $s(\mathcal{P})=0$.

The central difficulty in the mixed-state setting is that the key identity
used in the pure case,
\begin{equation}
    I(P(\ell)) + I(R : P(N-\ell)) = I(R:S),
    \label{eq:pureComplementary}
\end{equation}
no longer holds.  
For mixed initial states, the corresponding relation becomes
\begin{equation}
\label{eq:mixedIneq}
    I(R : P(\ell)) + I(R : P(N-\ell)) \;\le\; I(R : S),
\end{equation}
which is insufficient for guaranteeing the recoverability condition:
even if $I(R:P(\ell))$ is made small, Eq.~\eqref{eq:mixedIneq}
does not imply that $I(R:P(N-\ell))$ is large.
Thus, a direct analogue of the pure-state proof fails for mixed states,
and we cannot establish recoverability using mutual-information
bounds alone.\footnote{A detailed examination of the mixed-state
R\'enyi mutual information appears in Appendix~\ref{app:mixed}.}

Nevertheless, a complete characterization of mixed-state ramp schemes can
be obtained using a structural fact:any mixed-state encoding can be realized as a pure-state encoding
followed by discarding a subset of shares ~\cite{r7,r8}.  
Equivalently, any mixed state on $N$ players with entropy
$s(\mathcal{P})$ can be purified by introducing $s(\mathcal{P})$
additional qubits that were subsequently traced out.
This observation allows us to lift the pure-state theorem to the mixed
setting.

\begin{theorem}
\label{theorem:mixedstate}
A Haar-random scrambling unitary $U$ implements an approximate quantum
ramp secret sharing scheme with parameters
\[
\bigg(
b = \frac{N + s(\mathcal{P})}{2} - \epsilon,
\quad
g = \frac{N + s(\mathcal{P})}{2} + \epsilon
\bigg),
\]
where $s(\mathcal{P})$ is the von Neumann entropy of the players'
initial state.
\end{theorem}

\begin{proof}
Consider a purification of the mixed initial state of the players
$\mathcal{P}$.  
Let $N'$ denote the number of qubits in the purified system before any
discarding occurs.  
Tracing out $s = s(\mathcal{P})$ qubits yields the mixed state on
$N = N' - s$ players.

Applying Theorem~\ref{theorem:purestate} to the purified system:
it yields a pure approximate ramp scheme with parameters
\[
b_{\rm pure} = \frac{N'}{2} - \epsilon,
\qquad
g_{\rm pure} = \frac{N'}{2} + \epsilon.
\]
Now discard $s$ of the $N'$ output shares.
Discarding cannot increase the amount of information available to
an unauthorized set, nor can it improve recoverability for an
authorized set.  
Thus the ramp parameters $(b_{\rm pure},g_{\rm pure})$ remain valid as long
as $g_{\rm pure} \le N$, i.e.\ as long as every authorized set remains
contained in the retained system.

Since $N' = N + s(\mathcal{P})$, we obtain
\[
b = \frac{N + s(\mathcal{P})}{2} - \epsilon,
\qquad
g = \frac{N + s(\mathcal{P})}{2} + \epsilon,
\]
establishing the claimed mixed-state parameters.
\end{proof}

As consistency checks:
\begin{itemize}
\item If $s(\mathcal{P}) = 0$ (pure players), we recover the scheme of
Theorem~\ref{theorem:purestate}.
\item If $s(\mathcal{P}) = N$ (maximally mixed players), then
$g = b = N$, yielding a trivial scheme $((N,N))$ in which all $N$
players are required to reconstruct the secret.
\end{itemize}

The maximally mixed case can also be seen directly from the R\'enyi
mutual information.  
Substituting $s(\mathcal{P})=N$ into Eq.~\eqref{eq:R'enyiMutualInformation}
gives
\begin{equation}
       I^{(2)}(\ell)
       = \log_2\!\bigl(1 + 3 \cdot 4^{\,\ell - N}\bigr).
\end{equation}
At $\ell = N$, we obtain $I^{(2)} = 2$, the maximum possible for a
single-qubit secret, while for $\ell < N$ the mutual information decays
exponentially in $N-\ell$.
Thus only the full set of $N$ players can recover the secret.

Figure~\ref{fig:mixedScheme1} illustrates how discarding different numbers
of players (equivalently: increasing $s(\mathcal{P})$) transforms the pure
scheme into different mixed-state ramp schemes.  
By varying $s(\mathcal{P})$ and the approximation parameter $\epsilon$,
one can realize all possible ramp secret-sharing parameters.

\begin{figure}
    \centering
         \includegraphics[width=0.7 \textwidth]{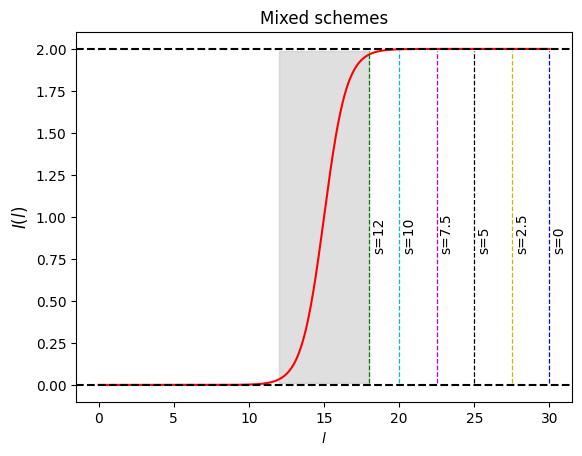}
   \caption{Mixed scheme can be thought of as a pure scheme with some players/systems discarded. $S= 0$ corresponds to the pure case where no players are discarded, while $s= 12$ corresponds to twelve players discarded. While the ramp parameters are the same, the total system size changes as the players are discarded. This formalism can also be thought of as a scheme where players have different entropies. By changing the entropy, one can then obtain all possible ramp schemes. }
   \label{fig:mixedScheme1}

\end{figure}

\subsection{Complexity Analysis}
In this section, we will perform the complexity analysis of our protocol. So far, we have assumed the Haar random unitaries for implementing the information scrambling circuit. Unfortunately, a simple counting argument shows that it requires an exponential number of gates, $\mathcal{O}(4^n)$, to implement the $n$ qubit Haar unitary \cite{knill1995approximation}. Furthermore, if we want to implement it with local $1$ and $2$ qubit gates, then it would take  $\mathcal{O}(n^22^{2n})$ gates \cite{Haar1}. This gives the impression that implementing this protocol efficiently is almost impossible. Fortunately, this is not true. For a unitary to be scrambling, it doesn't have to be a perfect Haar unitary. In fact, very simple operators can actually be pretty good scramblers. Both Haar-random and unitary design states, for $t \geq 2$, exhibit near-maximal entanglement. This implies that the state one obtains from t-design circuits \cite{Gross_2007} is information theoretically indistinguishable from a Haar scrambling circuit up to t moments \cite{ambainis2007quantum}. Since we used an information-theoretic argument to obtain ramp secret sharing schemes from Haar scrambling, for our purpose, t-design circuits are already enough. An example of a perfect 2-design circuit is provided by the ensemble of operators that form the Clifford group \cite{Dankert_2009}. This is a particularly useful result for us because Clifford circuits can be obtained by using only $\mathcal{O}(n^2)$ one and two-qubit gates, which means our proposed protocol can be performed efficiently, i.e., with a polynomial number of gates.

There are several comments now in order. Gottesmann Knill's theorem \cite{gottesman1998heisenberg} implies that to get any quantum advantage, one needs to use gates outside of the Clifford group, as a circuit involving only Clifford gates can be efficiently simulated by a classical computer. However, this is true only for computational tasks, and in several areas, such as quantum communication and quantum cryptography, protocols involving only Clifford group gates are very prominent. For example, quantum teleportation and superdense coding use only Clifford gates. For the NISQ Era, when errors are particularly prone, quantum cryptography protocols like the one we propose may be a promising quantum application, as Clifford gates are easier to correct errors compared to non-Clifford ones. The second comment is that while unitary design circuits are enough for our purpose, pseudorandom quantum states, on the other hand, are not. Quantum pseudorandom states are efficiently constructible states that nevertheless masquerade as Haar-random states to polynomial-time observers. Therefore, pseudorandom quantum states are not information-theoretically equivalent to Haar-random unitaries and are only equivalent from the perspective of observers with polynomial time complexity. In fact, it is possible to construct pseudoranodm states with only polylogarithmic entanglement entropy across an equipartition of the qubits \cite{aaronson2023quantum}. Therefore, these pseudorandom unitaries lack sufficient scrambling features to implement the secret sharing protocol.

Now that we have discussed the complexity of information scrambling, let us discuss the complexity of the descrambling. Cryptography involves both encoding and decoding secrets, which in our framework refers to scrambling and descrambling quantum information. Fortunately, quantum mechanics is unitary, and if we have the knowledge about how the information was scrambled, in principle, we should be able to descramble it.  It was shown that if $C$ is the complexity of the scrambling circuit, then one can construct the descrambling/decoding circuit of the complexity order $d_S C$, where $d_S$ is the Hilbert space dimension of the quantum secret \cite{yoshida2017efficient}. For our protocol, where scrambling is implemented using only a polynomial number of Clifford gates, efficient descrambling is also possible, that is, using only a polynomial number of gates. While it is true that one needs to be provided with the encoding unitary to perform the  decoding, one can sometimes do better. For example, one can perform the tomography to learn enough about the circuit and decode it. This is particularly true for the Clifford circuit, for which even without a scrambling unitary, we can perform tomography and learn enough to perform the decoding operation. In \cite{oliviero2022black,leone2023learning}, it was shown that this is possible even for certain families of non-Clifford circuits which are quasi-chaotic in nature. In particular, they showed that if the scrambling unitary $U_t$ is a $ t$-doped Clifford circuit, it is possible to learn the scrambling circuit using only $\mathcal{O}(\text{poly}(n)\exp (t))$ queries. Interestingly, the decoder can still be a family of Clifford gates only. This fails if the scrambling circuit is fully chaotic, and then it is not possible to learn the circuit efficiently. Being able to perform the full tomography doesn't affect our protocol, as in quantum secret sharing schemes, it is possible to allow the parties to have full knowledge of how encoding as well as decoding is performed. 

\begin{theorem}
The protocol, both encoding and decoding, can be implemented efficiently, i.e., using only polynomial number, $\mathcal{O}(n^2)$, one and two qubit gates. 
\end{theorem}

\section{Applications}

In this section, we will briefly discuss potential applications of our work, spanning from establishing secure quantum networks to evaluating cryptographic tasks in the NISQ Era.
\subsection{Information Scrambling protocol in Quantum Network and Security Implications}
\begin{figure}
    \centering
    \includegraphics[width=1.0\textwidth]{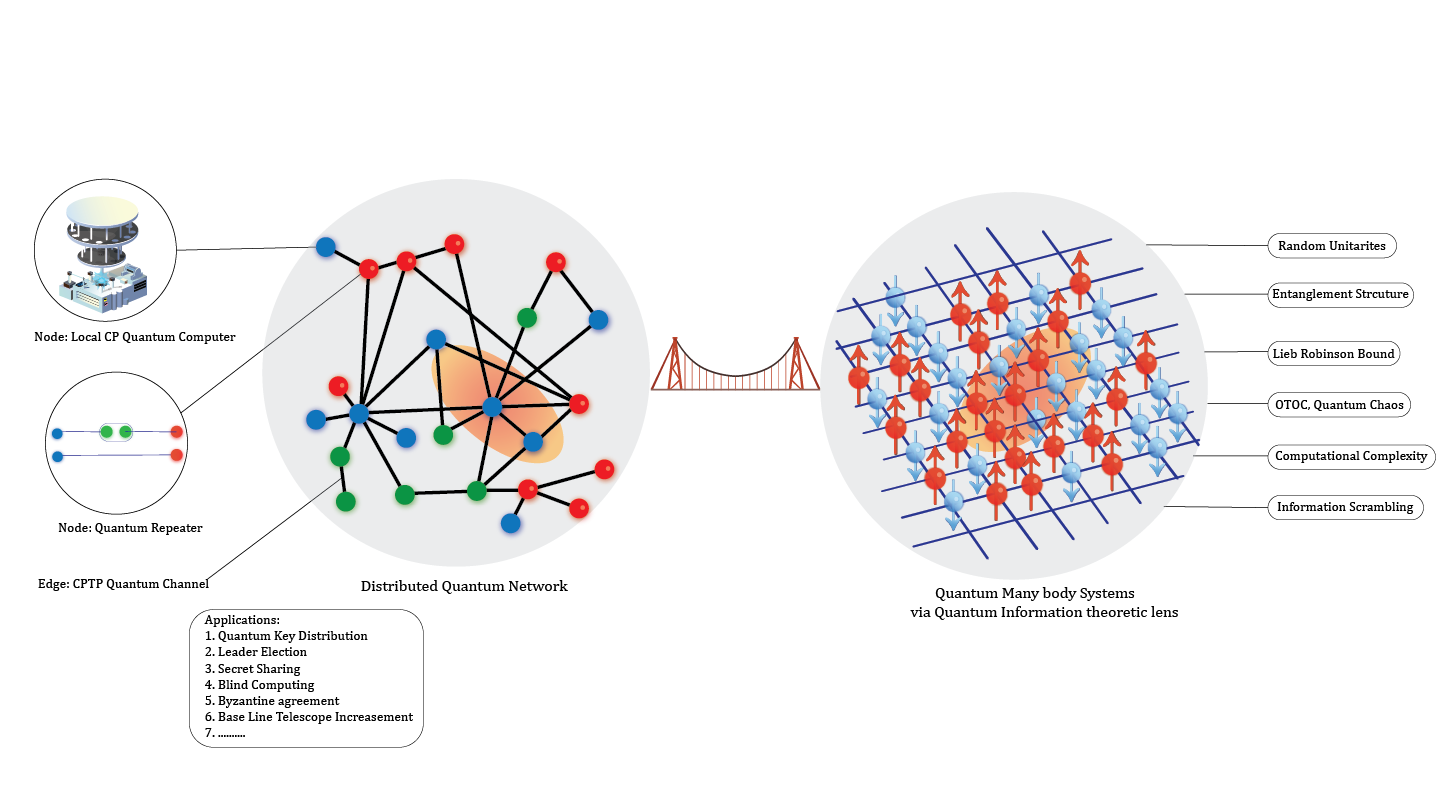}
    \caption{A bridge between distributed quantum networks and quantum many-body systems using the language of quantum
information theory. There are several interesting results obtained in the field of quantum many-body systems using quantum
information theory, such as Random unitaries, entanglement structure, quantum chaos, information scrambling, and phase transitions
 and so on. Thinking of quantum networks as a many-body system, similar techniques can be implemented to design
new protocols. For example, Quantum chaos could be used to model hacking of the quantum network, phase transitions in
entanglement structure can be used to study maximally entangled network clusters,}
    \label{fig:bridge}
\end{figure}
Classical communication networks serve as the foundation of our modern society. Alongside classical data, it is also possible to exchange quantum data over a long distance. These large-scale quantum networks \cite{qcn, wehner, singh2021quantum, adhikari2023quantuminformationspreadingscrambling} offer novel capabilities, such as key distribution, clock synchronization, and increasing the baseline of telescopes, among others. One crucial question in the field of quantum networks is to find killer applications to justify actually building it. In this context, the utilization of information scrambling protocols opens up exciting possibilities for exploring the unique features of quantum networks, leading to novel security implications.
We observed that by varying the purity of parties, one can construct all possible types of secret sharing schemes. Therefore, it is possible to construct a scrambling protocol such that the malicious party would need to get access to a significant fraction of the network to get any information
out of the network. This theme can be further extended to establish a connection between two distinct fields: distributed quantum networks and quantum many-body systems, as depicted in Figure \ref{fig:bridge}.

\subsection{Benchmarking and cryptography protocols in NISQ Computers}
In the NISQ Era \cite{Preskill_2018}, focus is often given to optimization and simulation problems, which, unfortunately, require many qubits. In contrast, cryptographic protocols require much fewer qubits,
such as quantum key distribution, quantum teleportation, quantum money, and superdense coding. One can implement our protocol in the NISQ Era with low-depth circuits, which can be used for benchmarking and finding novel quantum cryptography applications.

\subsection{Fundamental physics perspective}
Quantum information scrambling is emerging as an essential tool in exploring quantum many-body dynamics, randomization
and benchmarking, quantum computing protocols, the onset of quantum chaos, and even in black-hole physics \cite{Adhikari_2022, Adhikari_2023, Adhikari_2024, Chhetriya:2025ndi}. By establishing a connection between information scrambling and secret sharing, applying quantum secret-sharing tools can offer fresh insights into these areas. 
For example, employing the concept of dividing the quantum system into authorized and unauthorized parties allows us to tell exactly where the quantum information is located. When modeling an old black hole as a scrambling unitary \cite{Page_1993, Hayden_2007}, one can deduce that after emitting $d$ Hawking radiation, where $d$ represents the size of the diary fallen into the black hole, the black hole essentially becomes an unauthorized set and can no longer retain any information about the diary. All the information about the diary should be in Hawking radiation, as it is an authorized set.

\section{Conclusion}
We have demonstrated that Haar scrambling features possess quantum secret sharing properties. Rather than having the perfect threshold
scheme, the secret sharing feature of Haar scrambling is closer to that of approximate ramp secret sharing schemes. By changing the purity of the initial state, one can then achieve all possible schemes. Furthermore, using complexity theoretic arguments,
we have shown that the protocol can be implemented efficiently just by using a polynomial circuit. Besides this, our work has
several potential applications, such as building secure quantum internet protocols, benchmarking and cryptography tasks, and
lessons for fundamental physics. There are several interesting open questions one could explore. One could work out the 
applications in full detail and find other applications too. One could also extend this work to the noisy cases. While
our model has inherent error correcting features, as one could still reconstruct the secret while losing qubits, provided they are
authorized, it would be very hard to decode the secret out if we don’t know the scrambling unitary, which would be the case
for the  noisy regime. Therefore, analyzing this model in the presence of noise is important. Another important question is of a more
theoretical nature such as finding the relation between $l$- scrambling and OTOCs, building up theory of quantum ramp secret
sharing, and exploring secret sharing features of other quantum systems

%% file: appendix.tex
\section{Toy model of Figure \ref{fig:informationScrambling2}} 
\label{app:toymodel}


In Figure \ref{fig:informationScrambling2}, we start with a system of $N$ qubits that has a total Hilbert space of $2^N$. The time evolution of an initial state $\ket{\psi}$ can be described by a unitary operator $U(t)$ as follows:
\begin{equation}
\ket{\psi (t)} = U(t) \ket{\psi}
\end{equation}
In order to have an information theoretic setting, we assume that Alice and Bob aim to communicate via the scrambling unitary $U$. Additionally, Alice controls the first qubit $q_A$, while Bob has access to a set of qubits $B$ (see Fig. \ref{fig:informationScrambling2}). To describe a general scenario, we consider the initial state to be mixed:
\begin{equation}
\begin{split}
    \rho_0 &= \ket{0} \bra{0} \otimes \rho \\
    \rho_1 &= \ket{1} \bra{1} \otimes \rho
\end{split}    
\end{equation}
Here, $\ket{0}$ and $\ket{1}$ denote two possible orthogonal states of Alice's qubit $q_A$, and $\rho$ is the density matrix of the system excluding Alice's qubit. Alice's qubit is entangled with an external reference system $R$, and an external system called "Memory" purifies the density matrix $\rho$ (as depicted in Fig. \ref{fig:informationScrambling2}). We represent this purified state as $\ket{\sqrt{\rho}}$, and by tracing out the memory system, we can retrieve the density matrix $\rho$:
\begin{equation}
\rho = \text{Tr}_{\text{MEM}} \ket{\sqrt{\rho}} \bra{\sqrt{\rho}}
\end{equation}
The initial state of the system, including the reference, system and the memory is:
\begin{equation}
     \ket{\boldsymbol{\psi}} = \frac{1}{\sqrt{2}}(\ket{0}_R\ket{0}_{q_A} +\ket{1}_R\ket{1}_{q_A}) \ket{\sqrt{\rho}}
\end{equation}
Since the time evolution unitary doesn't act on reference and memory, the time evolved state is given by:
\begin{equation}
\begin{split}
  \ket{\boldsymbol{\psi}(t)}  &= \frac{1}{\sqrt{2}} (\ket{0}_R U_S \otimes I_\text{MEM} \ket{0,\sqrt{\rho}}  +  \ket{1}_R U_S \otimes I_\text{MEM} \ket{1,\sqrt{\rho}} )
\end{split}    
\end{equation}

\section{Haar Integrals}
\label{app:HaarIntegrals}

An integral over the unitary group $U$ ($2^n \times 2^n$) of the matrix function $f(U)$ with respect to the Haar measure can be represented as:
\begin{equation}
I = \int dU f(U)
\end{equation}
The defining property of the Haar measure is left- (respectively, right-) invariance with respect to shifts via multiplication. If $V$ is a fixed unitary, then:
\begin{equation}
\int dU f(UV) = \int d(U'V^\dagger) f(U') = \int dU'f(U')
\end{equation}


In order to properly analyze the setting, we divide the input state as $AB$ and output state $CD$. Furthermore, the unitary $U$ is taken to be random $2^n \times 2^n$ matrix taken from Haar ensemble. The Haar
average lets us consider expectations over a number of unitary matrices and is non-zero
only when the number of $U$ equals the number of $U^\dagger$.
For example with the case of two $U$s and two $U^\dagger$s, we get:
\begin{equation}
\label{eq:haarintegral}
\begin{aligned}
       \int dU U_{i_1j_1}U_{i_2j_2}U^*_{i'_1j'_1}U^*_{i'_2j'_2} = \frac{1}{2^{2n}-1}&  (\delta_{i_1 i'_1}\delta_{i_2 i'_2}\delta_{j_1 j'_1}\delta_{j_2 j'_2} + \delta_{i_1 i'_2}\delta_{i_2 i'_1}\delta_{j_1 j'_2}\delta_{j_2 j'_1})\\
       &- \frac{1}{2^n(2^{2n}-1)} (\delta_{i_1 i'_1}\delta_{i_2 i'_2}\delta_{j_1 j'_2}\delta_{j_2 j'_1}  + \delta_{i_1 i'_2}\delta_{i_2 i'_1}\delta_{j_1 j'_1}\delta_{j_2 j'_2})
\end{aligned}
\end{equation}
Let $d= d_Ad_B = 2^a 2^b= d_Cd_D = 2^c2^d = 2^{a+b} = 2^{c+d} = 2^n$, then $d = 2^n$.

This formula will let us compute the average over the trace of the square of the density matrix $\rho_{AC}$:
\begin{equation}
    \int dU \text{Tr}(\rho^2_{AC}) = \frac{1}{2^{2n}} \int dU U_{klmo}U^*_{k'lm'o}U_{k'l'm'o'} U^*_{kl'mo'}
\end{equation}
where, $k = 1,...,2^a$ are $A$ indices, $l = 1,......,2^b$ are $B$ indices, $m = 1,.....,2^c$ are $C$ indices, and $o = 1,......, 2^d$ are $D$ indices. 
Now, using equation \ref{eq:haarintegral},
\begin{equation}
\begin{aligned}
    \int dU \text{Tr}(\rho^2_{AC}) &= \frac{1}{d^2}\left( \frac{1}{d^2-1}(d_Ad_B^2d_Cd_D^2 + d_A^2d_Bd_C^2d_D) - \frac{1}{d(d^2-1)}(d_Ad_B^2d_C^2d_D + d_A^2d_Bd_Cd_D^2) \right)       
\end{aligned}
\end{equation}
We can approximate above equation as:
\begin{equation}
    \int dU \text{Tr}(\rho^2_{AC}) = d_A^{-1}d_C^{-1} + d_B^{-1}d_D^{-1} - d^{-1}d_A^{-1}d_D^{-1} - d^{-1}d_B^{-1}d_C^{-1} 
\end{equation}
From these quantities, it is possible to obtain the expression for R\'enyi enropies, which allow us to obtain bound for other information theoretic quantities. 
\section{Out of time ordered correlators (OTOCs)}
\label{app:OTOCs}
To illustrate how correlators can be used to study information spreading, we consider a toy model of Alice and Bob communicating via a system of $N$ qubits. Alice controls the first qubit $q_A$, and Bob has access to a set of qubits $B$. Alice wants to send a bit $a \in {0,1}$ to Bob, which she does by flipping her qubit with the $\sigma_{q_A}^x$ operator depending on $a$. Bob then measures his qubits using $O_B$ after the system evolves for some time $t$.

 
 Alice's qubit flip affects Bob's expectation value, with $\langle O_B \rangle _0$ and $ \langle O_B \rangle _1$ denoting the respective values. By assuming that Alice and Bob repeat the measurements many times, the difference in expectation values can be bounded using the Cauchy-Schwarz inequality. This difference is bounded by the operator norm of the commutator $[\sigma_{q_A}^x (-t), O_B]$, where $\sigma^x(-t) = U \sigma^x U^\dagger$ is a Heisenberg operator. 
A simple and intuitive way to understand this phenomenon is as follows: at $t=0$, the operator $\sigma^x$ only acts on Alice's qubit and has no overlap with Bob's qubit. Therefore, the commutator between $\sigma^x(-t)$ and Bob's observable $O_B$ is initially zero. As time progresses, the support of $\sigma^x(-t)$ grows as a Heisenberg operator, and the overlap with $O_B$ increases. This results in a non-zero commutator, indicating that Bob can now gain information about Alice's qubit. This spreading of information can be bounded rigorously using the Lieb-Robinson bound, which provides an upper bound on the rate at which information can travel through a quantum system.

Now we will define quantum chaos using out-of-time order correlators (OTOCs) \cite{larkin, Maldacena_2016, Sekino_2008}, following the previous discussion of a toy model. OTOCs can be seen as the quantum version of the classical observation that the sensitivity to initial conditions can be quantified by a Poisson bracket: ${q(t),p} = \partial q(t)/ \partial q(0)$, where $q$ and $p$ are a conjugate pair. In a classically chaotic system, this quantity would be proportional to $ e^{\lambda t}$, where $\lambda$ is the Lyapunov exponent. The exponential growth of OTOCs is then defined as the thermal average of the square of the commutator $[\hat{q}(t),\hat{p}]$, obtained by quantizing ${q(t),p}$.

Generally, we consider a function involving the commutation of arbitrary operators (out-of-time ordered commutator) $C(t)$ given by:
\begin{equation}
\label{eq:otoc commutator}
\begin{aligned}
C(t) &= \langle [V(0), W(t)]^\dagger [V(0), W(t)] \rangle \\
&= 2 \text{Re} \langle W^\dagger(t)W(t)V(0)V^\dagger(0)\rangle - 2 \text{Re} \langle W^\dagger(t)V^\dagger(0)W(t)V(0) \rangle
\end{aligned}
\end{equation}
Here, $O(0) = O$, $O(t) = U^\dagger O U$, and $\langle O \rangle = \text{Tr} O \rho$ denotes an average of an operator $O$ over some set of states. Usually, the average is taken over a thermal state, $\rho \propto d^{-1}I$, where $d$ is the Hilbert space dimension. The first term in Eq. \ref{eq:otoc commutator} is the time-ordered correlator and equilibrates at the thermalization time, while the second term is the out-of-time ordered correlator (OTOCs), which we refer to as $F(t)$, and characterizes the quantum chaotic effect:
\begin{equation}
F(t) = \langle W^\dagger(t)V^\dagger(0)W(t)V(0) \rangle.
\end{equation}
OTOCs measure the overlap between two states, $F(t) = \langle \psi_B | \psi_A \rangle $, where the states are given by:
\begin{equation}
\begin{aligned}
\ket{\psi_A} &= U^\dagger W U V \ket{\psi} = W(t) V \ket{\psi}, \\
\ket{\psi_B} &= V U^\dagger W U \ket{\psi} = V W(t) \ket{\psi}.
\end{aligned}
\end{equation}

\begin{figure}
    \centering
    \includegraphics[width=1.0\textwidth]{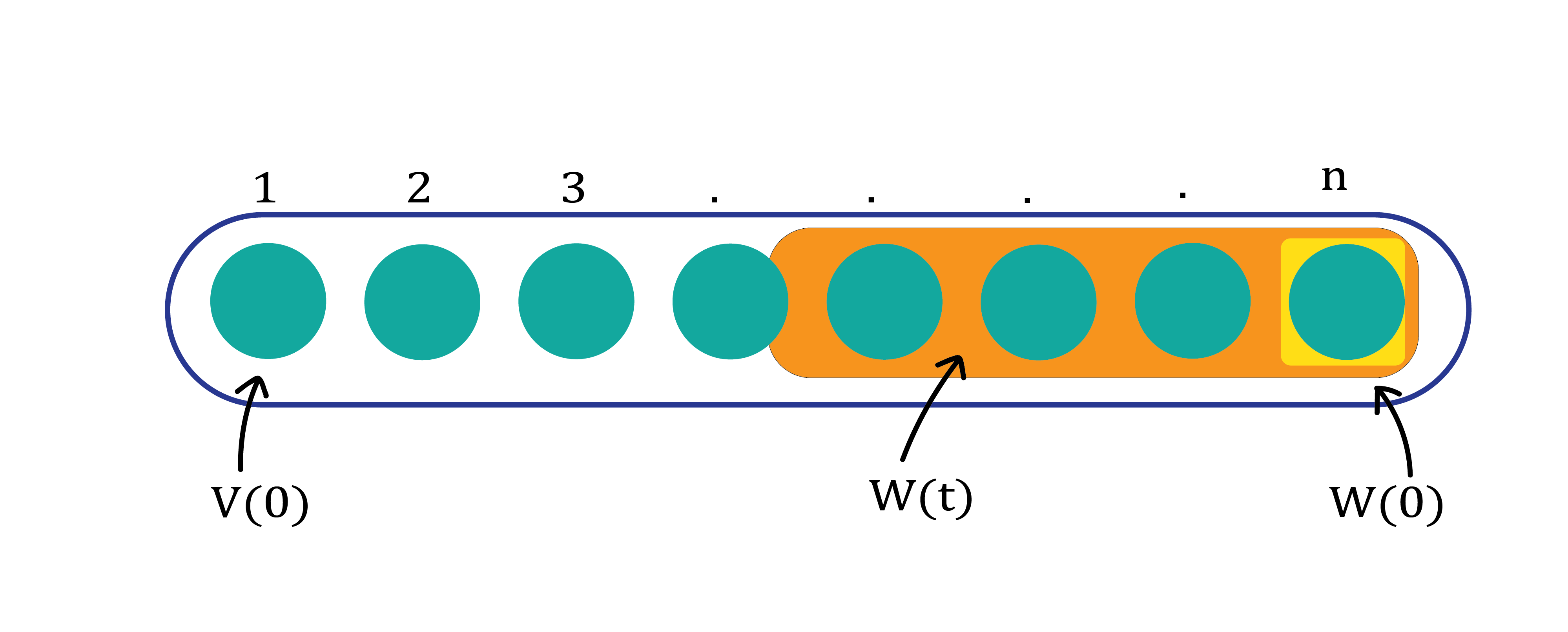}
    \caption{ Evolution of local operator $W(0)$ under local Hamiltonian given by $W(t)$. As time increases, $W(t)$ is supported on larger region of the system. Commutation between $V(0)$ and $W(t)$ is non-zero only when region supported by $W(t)$ and $V(0)$ overlaps.}
    \label{fig:chaoticScrambling}
\end{figure}
One can think of $F(t)$ as a four-point correlation function that probes the way in which (local) perturbations inhibit the cancellation between forward and backward evolution, and it diagnoses the spread of quantum information by measuring how quickly two commuting operators fail to commute. Since $W(0)$ and $V(0)$ are supported on non-overlapping subsystems, we have $ [W(0),V(0)] = 0$ thus
\begin{equation} F(t) = \langle W^\dagger(t)V^\dagger(0)W(t)V(0) \rangle = \langle W^\dagger(t)W(t)V^\dagger V(0) \rangle = 1. \end{equation} 
Under a chaotic time evolution with Hamiltonian $H$, a local operator $W(0)$ will evolve into a complicated operator with an expansion as a sum of products of many local operators: \begin{equation} W(t) = e^{iHt}We^{-iHt} = W + it[H,W] - \frac{t^2}{2!} [H,[H,W]] - \frac{it^3}{3!} [H,[H,[H,W]]]+ \ldots, \end{equation} where the ellipsis denotes higher-order terms. Assuming a generic Hamiltonian $H$, the time-evolved operator $W(t)$, starting from a local operator $W(0)$, is supported on a larger region and hence does not commute with the initially non-overlapping operator $V(0)$, i.e., $[W(t),V(0)] \neq 0$. Consequently, the four-point correlation function $F(t)$, which measures the spread of quantum information, starts to decay and eventually goes to zero. This decay behavior is depicted in the cartoon representation shown in Figure \ref{fig:chaoticScrambling}. If this decay behavior occurs, we refer to the unitary as chaotic scrambling. Intuitively, this implies that over time, a chaotic scrambling unitary requires access to a larger region to access the initially localized information.
\section{$((2,3))$ quantum threshold scheme}

\label{app:example}
Here a dealer share an arbitrary qutrit among three parties $V_{2,3}:\mathcal{C}^3 \xrightarrow{} \mathcal{C}^3 \otimes \mathcal{C}^3 \otimes \mathcal{C}^3$ defined as:
\begin{equation}
\begin{split}
        V_{2,3}(\alpha \ket{0} + \beta \ket{1} + \gamma \ket{2}) = \frac{1}{\sqrt{3}}(&\alpha (\ket{000} + \ket{111} + \ket{222} ) + \\
        & \beta ( \ket{012} + \ket{120} + \ket{201}) + \\
        & \gamma (\ket{021} + \ket{102} + \ket{210}))
\end{split}
\end{equation}
Indeed, $V_{2,3}$ is an isometry. Each resulting qutrit is taken as a share. From a single share, no information can be obtained about the secret as it is totally mixed and therefore it is independent about the secret.
\begin{equation}
\rho_1 = \frac{1}{3}(\ket{0}\bra{0} + \ket{1}\bra{1} + \ket{2}\bra{2})
\end{equation}
However, the secret can be perfectly reconstructed from any two of the three shares. For example, if we have two first shares, then the secret can be reconstucted as:
\begin{equation}
   (\alpha \ket{0} + \beta \ket{1} + \gamma \ket{2}) (\ket{00} + \ket{12} + \ket{21})
\end{equation}
So, the first qutrit now contains the secret. Here, the reconstruction procedure can be done by adding the value of the first share to the second (modulo three), and then adding the value of the second share to the first. For the other cases, we can follow the similar procedure because of it's symmetric nature. 


\section{Information theoretic analysis for the mixed states}
\label{app:mixed}

Let $P(\ell)$ be an arbitrary region of size $\ell$ in the system after the
Haar scrambling unitary is applied, and let $I(P(\ell))$ denote the mutual
information between the original secret (one qubit) and this region.  
Using the R\'enyi-based upper bound on the quantum mutual information
derived in Eq.~\eqref{eq:R'enyiMutualInformation}, we obtain an inequality
of the form
\begin{equation}
    I(P(\ell))
    \;\le\;
    F\bigl(N, s(\mathcal{P}), \ell\bigr),
    \label{eq:mixedRenyiBound}
\end{equation}
where $F$ is an explicit function depending on the total number of players
$N$, the initial entropy $s(\mathcal{P})$ of the players, and the subsystem
size~$\ell$.  Concretely, Eq.~\eqref{eq:R'enyiMutualInformation} yields
\begin{equation}
    I(P(\ell))
    \;\le\;
    1 + \log_2\!\left(
        2 - \frac{3\bigl(1 - 2^{2\ell - 2N}\bigr)}
        {2 + \bigl(2^{-s(\mathcal{P})} - 2^{-N}\bigr)\,4^{\ell - N/2}}
    \right),
\end{equation}
which reduces to the pure-state expression when $s(\mathcal{P}) = 0$.

To connect this to ramp secret sharing, it is natural to consider the values
of $\ell$ near the ``transition'' scale
\[
    \ell_\star = \frac{N + s(\mathcal{P})}{2},
\]
as the center of the ramp region for mixed-state schemes.  In particular, for
\begin{equation}
    \ell = \frac{N + s(\mathcal{P})}{2} - \epsilon,
\end{equation}
and fixed $N$ and $s(\mathcal{P})$, the function
$F\bigl(N, s(\mathcal{P}), \ell\bigr)$ can be made arbitrarily small by
choosing $\epsilon$ appropriately.  That is, for any target accuracy
$\gamma > 0$, there exists $\epsilon = \epsilon(\gamma)$ such that
\begin{equation}
    I\!\left(P\!\left(\tfrac{N + s(\mathcal{P})}{2} - \epsilon\right)\right)
    \;\le\; \gamma.
    \label{eq:mixedSecrecy}
\end{equation}
This establishes a relaxed secrecy condition for subsets of size
$\ell = \frac{N + s(\mathcal{P})}{2} - \epsilon$: they have vanishingly
small correlation with the reference as $\gamma \to 0$.

However, in the mixed-state scenario, the complementary mutual information
no longer satisfies the pure-state equality
\eqref{eq:pureComplementary}.  
Instead, one has only the inequality
\begin{equation}
    I(R : P(\ell)) + I\bigl(R : P(N-\ell)\bigr)
    \;\le\; I(R:S),
    \label{eq:mixedComplementary}
\end{equation}
which follows from strong subadditivity and the fact that the global state
is no longer pure on $RP(N)$ alone.  For
\[
    \ell = \frac{N + s(\mathcal{P})}{2} - \epsilon,
\]
the relaxed secrecy bound \eqref{eq:mixedSecrecy} gives
\begin{equation}
    I\bigl(R : P(\ell)\bigr) \le \gamma,
\end{equation}
and the inequality \eqref{eq:mixedComplementary} then implies only
\begin{equation}
    I\bigl(R : P(N-\ell)\bigr)
    \;\le\;
    I(R:S) - I(R:P(\ell))
    \;\le\;
    I(R:S),
    \label{eq:mixedUpper}
\end{equation}
which is a trivial upper bound and provides no nontrivial lower
bound on $I\bigl(R : P(N-\ell)\bigr)$.  In other words, from the direct
R\'enyi-based analysis we obtain
\begin{equation}
    I\bigl(R : \text{unauthorized region}\bigr) \approx 0,
\end{equation}
but we cannot conclude that
\begin{equation}
    I\bigl(R : \text{complementary region}\bigr) \approx 2
\end{equation}
purely from Eq.~\eqref{eq:mixedComplementary}.

Thus, this information-theoretic approach suffices to verify a
secrecy condition for the mixed-state case, namely, that suitably
small subsets have negligible mutual information with the secret, but it
fails to establish the recoverability condition: we cannot show that
large subsets necessarily contain almost all the information about the
secret.  Consequently, the direct analysis based solely on
Eq.~\eqref{eq:mixedRenyiBound} and the inequality
\eqref{eq:mixedComplementary} does not yield a complete mixed-state quantum
ramp secret-sharing scheme.

This motivates the more structural approach adopted in
Section~\ref{theorem:mixedstate}, where mixed-state schemes are obtained
from pure ones via purification and discarding of shares.  That approach
provides both secrecy and recoverability and leads to the ramp parameters
\[
    b = \frac{N + s(\mathcal{P})}{2} - \epsilon,
    \qquad
    g = \frac{N + s(\mathcal{P})}{2} + \epsilon.
\]